\documentclass[a4paper]{article}
\usepackage{amsmath,amssymb,amsthm,latexsym,pslatex}
\usepackage{graphics,pspicture,epsfig}
\usepackage{a4wide,multirow,enumerate}

\theoremstyle{plain}
\newtheorem{theorem}{Theorem}
\newtheorem{lemma}{Lemma}  
\newtheorem{fact}{Fact}
\newtheorem{proposition}{Proposition} 
\newtheorem{corollary}{Corollary}

\theoremstyle{definition}
\newtheorem{definition}{Definition}
\theoremstyle{remark}

\newcommand{\braced}[1]{{\left\{#1\right\}}}
\newcommand{\fst}{\textit{fst}}
\newcommand{\lst}{\textit{lst}}

\begin{document}

\title{Faster Small-Constant-Periodic Merging Networks}

\author{Marek Piotr\'ow\\
Institute of Computer Science, University of Wroc\l aw,\\
ul.~Joliot-Curie~15, PL-50-383 Wroc\l aw, Poland\\
email: Marek.Piotrow@ii.uni.wroc.pl}
\date{}
\maketitle

\begin{abstract} 
We consider the problem of merging two sorted sequences on a comparator network
that is used repeatedly, that is, if the output is not sorted, the network is
applied again using the output as input. The challenging task is to construct such
networks of small depth (called a period in this context). In our previous paper
{\em Faster 3-Periodic Merging Network} we reduced twice the time of merging on
$3$-periodic networks, i.e. from $12\log N$ to $6\log N$, compared to the first
construction given by Kuty{\l}owski, Lory{\'s} and Oesterdikhoff. Note that
merging on $2$-periodic networks require linear time. In this paper we extend our
construction, which is based on Canfield and Williamson $(\log N)$-periodic
sorter, and the analysis from that paper to any period $p \ge 4$. For $p\ge 4$ our
$p$-periodic network merges two sorted sequences of length $N/2$ in at most
$\frac{2p}{p-2}\log N + p\frac{p-8}{p-2}$ rounds. The previous bound given by
Kuty{\l}owski at al. was $\frac{2.25p}{p-2.42}\log N$. That means, for example,
that our $4$-periodic merging networks work in time upper-bounded by $4\log N$ and
our $6$-periodic ones in time upper-bounded by $3\log N$ compared to the
corresponding $5.67\log N$ and  $3.8\log N$ previous bounds. Our construction is
regular and follows the same periodification schema, whereas some additional
techniques were used previously to tune the construction for $p \ge 4$. Moreover, 
our networks are also periodic sorters and tests on random permutations
show that average sorting time is closed to $\log^2 N$.
\end{abstract}



\section{Introduction}  
\label{intro}

Comparator networks are probably the simplest, comparison-based parallel model
that is used to solve such tasks as sorting, merging or selecting \cite{k}. Each
network represents a data-oblivious algorithm, which can be easily implemented in
other parallel models and hardware. Moreover, sorting networks can be applied in
secure, multi-party computation (SMC) protocols. They are also used to encode
cardinality constrains to propositional formulas \cite{ano} and are strongly
connected with switching networks \cite{l}. The most famous constructions of
sorting networks are Odd-Even and Bitonic networks of depth $\frac{1}{2}\log^2 N$
due to Batcher \cite{b} and AKS networks of depth $O(\log N)$ due to Ajtai, Komlos
and Szemeredi \cite{aks}. The long-standing disability to decrease a large
constant hidden behind the asymptotically optimal complexity of AKS networks to a
practical value has resulted in studying easier, sorting-related problems, whose
optimal networks have small constants. For a review on merging networks and sorting 
network see, for example, Knuth \cite{k}.

A comparator network consists of a set of $N$ registers, each of which can store
an item from a totally ordered set, and a sequence of comparator stages.  Each
stage is a set of comparators that connect disjoint pairs of registers and,
therefore, can work in parallel (a comparator is a simple device that takes a
contents of two registers and performs a compare-exchange operation on them: the
minimum is put into the first register and the maximum into the second one).
Stages are run one after another in synchronous manner, hence we can consider the
number of stages as the running time. The size of a network is defined to be the
total number of comparators in all its stages.

A network $A$ consisting of stages $S_1,S_2,\ldots,S_d$ is called $p$-perio\-dic
if $p<d$ and for each $i$, $1\le i\le d-p$, stages $S_i$ and $S_{i+p}$ are
identical.  A periodic network can be easier to implement, because one can use the
first $p$ stages in a cycle: if the output of $p$-th stage is not correct (sorted,
for example), the sequence of $p$ stages is run again. In pure oblivious context,
such computations are stopped after a predefined number of passes. We can also
define a $p$-periodic network just by giving the total number of stages and a
description of its first $p$ stages. A challenging task is to construct a family
of small-periodic networks for sorting-related problems with the running time
equal to, or not much greater than that of non-periodic networks.

Dowd et al.\ \cite{dpsr} gave the construction of $\log N$-periodic sorting
networks of $N$ registers with running time of $\log^2 N$. Bender and Williamson
introduced a large class of such networks \cite{bw}.  Kuty{\l}owski et al.\
\cite{klow} introduced a general method to convert a non-periodic sorting network
into a 5-periodic one, but the running time increases by a factor of $O(\log N)$
during the conversion. For simpler problems such as merging or correction there
are constant-periodic networks that solve the corresponding problem in
asymptotically optimal logarithmic time \cite{klo,o,p}. In particular,
Kuty{\l}owski, Lory{\'s} and Oesterdikhoff \cite{klo} have given a description of
$3$-periodic network that merges two sorted sequences of $N$ numbers in time
$12\log N$ and a similar network of period $4$ that works in $5.67\log N$.  They
sketched also a construction of merging networks with periods larger than 4 and
running time decreasing asymptotically to $2.25\log N$. Note that $2$-periodic
merging networks require linear time.

In this paper we extend our construction from \cite{mpi14} of a new family of
$3$-periodic merging networks, which is based on Canfield and Williamson $(\log
N)$-periodic sorter \cite{cw}, and the underlying analysis to any period $p \ge
4$. For $p\ge 4$ our $p$-periodic network merges two sorted sequences of length
$N/2$ in at most $\frac{2p}{p-2}\log N + p\frac{p-8}{p-2}$ rounds. The previous
bound given by Kuty{\l}owski at al. \cite{klo} was $\frac{2.25p}{p-2.42}\log N$.
That means, for example, that our $4$-periodic merging networks work in time
upper-bounded by $4\log N$ and our $6$-periodic ones in time upper-bounded by
$3\log N$, compared to the corresponding $5.67\log N$ and  $3.8\log N$ previous
bounds. Our construction is regular and follows the same periodification schema as
we used for $3$-periodic merging networks, whereas some additional techniques were
used previously to tune the construction for $p \ge 4$. Increasing $p$ further, the
multiplicative constant decreases approaching 2. The construction is pretty
simple, but its analysis is quite complicated. 

The advantage of constant-periodic networks is that they have pretty simple
patterns of communication links, that is, each node (register) of such a
network can be connected only to a constant number of other nodes. Such
patterns are easier to implement, for example, in hardware.  Moreover, a node
uses these links in a simple periodic manner and this can save control login
and simplify timing considerations. We can also easily implement an early
stopping property with $p$-periodic networks: if none of the comparators
exchanged values in the last $p$ stages, we could stop the computation. Since
our networks are also periodic sorters, we have used this property to measure
sorting times on random permutations and the results are quite surprising: the
average sorting time of N items is closed to $\log^2 N$. Results are 
presented in Section \ref{sec4}.

The paper is organized as follows. In Section \ref{sec2} we introduce
a new periodification scheme, define our new family of $p$-periodic
merging networks and give the main theorem. Section \ref{sec3} is
devoted to its proof, where we order the set of registers into a
matrix and analyse the behaviour of our network by tracing the numbers
of ones in its columns.

\section{Periodic merging networks}\label{sec2} 

Our merging networks are based on the Canfield and Williamson
\cite{cw} $(\log N)$-periodic sorters. In the following proposition we recall the 
definition of the networks and their merging/sorting properties (see also Fig. 
\ref{anotherCW}). Recall that $[i:j]$ denotes a comparator connecting registers 
$i$ and $j$.
\begin{proposition}(see \cite{cw}) For $k\ge 1$ let $S_1 = \braced{[2i:2i+1]:
i=0,1,\ldots,2^{k-1}-1}$ and for $j=1,\ldots,k-1$ let $S_{j+1} =
\braced{[2i+1:2i+2^{k-j}]: i=0,1,\ldots,2^{k-1}-2^{k-j-1}-1}$. Let $CW_k =
S_1,\ldots,S_k$ be a network of $N_k=2^k$ registers numbered $0, \ldots, N_k-1$.
Then (i) if two sorted sequences of length $2^{k-1}$ are given in registers with
odd and even indices, respectively, then $CW_k$ is a merging network; (ii) $CW_k$
is a $k$-pass periodic sorting network.
\end{proposition}
\begin{figure}
\centering \epsfig{file=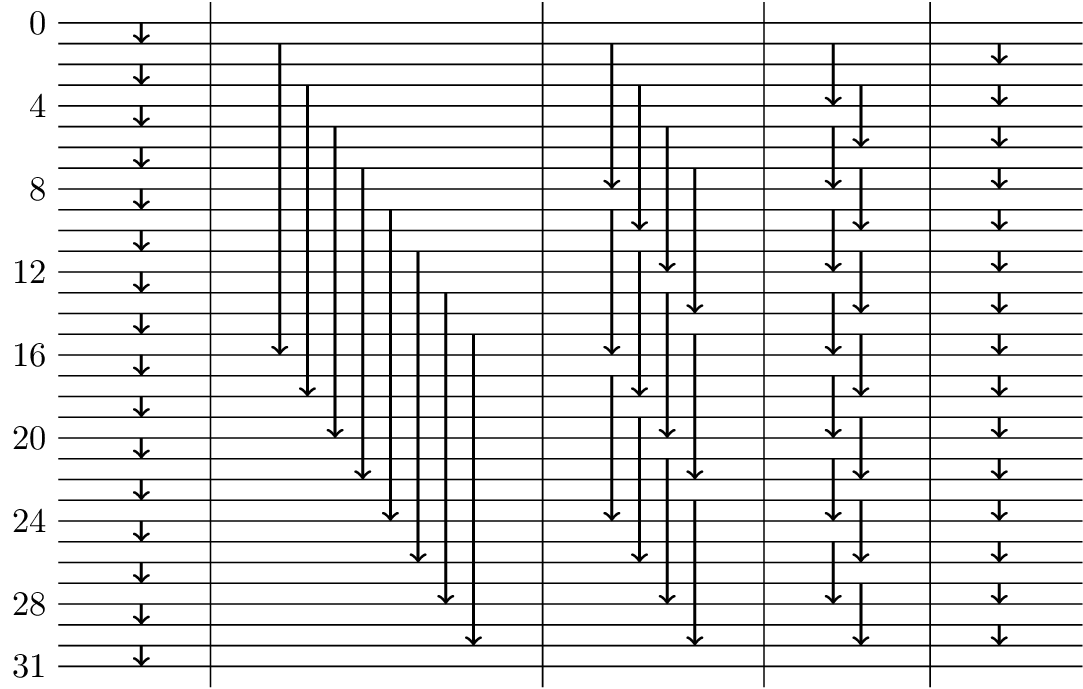, height=1.35in} 
\hspace{24mm}
\epsfig{file=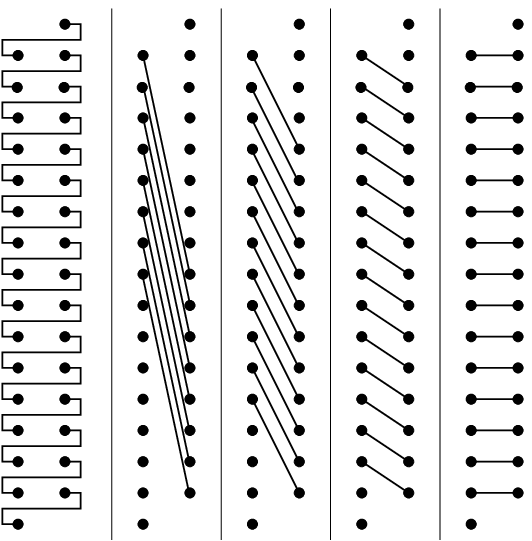, width=2in, height=1.35in} 
\caption{The Canfield and Williamson $(\log N)$-periodic sorter $CW_5$, where
$N=32$. On the left side, registers and comparators are represented by horizontal
lines and arrows, respectively. On the right side, registers and comparators are
represented by dots and edges, respectively. Stages are separated by vertical
lines.} 
\label{mergeCW}\label{anotherCW}
\end{figure}
We would like to implement a version of this network as a $p$-periodic comparator
network. We begin with the definition of an intermediate construction $P^p_k$
which structure is similar to the structure of $CW_k$. Then we transform it to
$p$-periodic network $M^p_k$. Observe that in any $N$-register merging network we
must have all {\em short} comparators $[i:i+1]$, $0 \le i < N-1$, and consecutive
short comparators $[i-1:i]$ and $[i:i+1]$ must be in different stages.  The idea
is to replace each register $i$ in $CW_k$ (except the first and the last ones)
with a sequence of $\lceil\frac{k-2}{p-2}\rceil$ consecutive registers, move the
endpoints of $i$-th group of $p-2$ {\em long} comparators one register further or
closer depending on the parity of $i$ and insert between each group of stages
containing long comparators a stage with short comparators joining the endpoints
of those long ones. The result is depicted in Fig. \ref {merge62}. In this way, we
obtain a network in which each register is used in at most $p$ consecutive stages.
Therefore the network $P^p_k$ can be packed into the first $p$ stages and used
periodically to get the desired $p$-periodic merging network. %
\begin{figure}
\centering
\epsfig{file=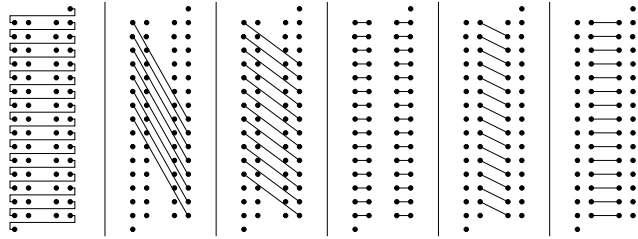, height=1.23in}
\caption{$P^4_5$ as an implementation of $CW_5$. Registers and comparators
  are represented by dots and edges, respectively. Stages are separated
  by vertical lines. Stages with short horizontal comparators are
  inserted between stages with long comparators.}
\label{merge62}
\end{figure}

A comparator $[i:j]$ is {\em standard} if $i<j$. All networks defined in this paper 
are built only of standard comparators. For an $N$-register
network $A = S_1,S_2,\ldots,S_d$, where $S_1,S_2,\ldots,S_{d}$ denote
stages, and for an integer $j\in\{1,\ldots,N\}$, we will use the
following notations:
\begin{align*}
   \fst(j,A) &= \min\braced{1\le i\le d: j \in regs(S_i)}\\
   \lst(j,A) &= \max\braced{1\le i\le d: j \in regs(S_i)}\\
   delay(A)  &= \max_{j\in\{1,\ldots,N\}}\braced{\lst(j,A)-\fst(j,A)+1}
\end{align*}
where $regs(\{[i_1:j_1], \ldots, [i_r:j_r]\})$ denotes the set
$\{i_1, j_1, \ldots, i_r, j_r\}$.

Let us define formally the new family of merging networks.  For each
$k\ge p\ge 4$ we would like to transform the network $CW_k$ into a new
network $P^p_k$. 
\begin{definition}\label{defMk}
  Let $n_k=2^{k-1}-1$ be one less than the half of the number of registers in
  $CW_k$, $b^p_k=2\lceil\frac{k-2}{p-2}\rceil$ and $D^p_k = k - 1 +
  \frac{b^p_k}{2}$. The number of registers of $P^p_k$ is defined to be 
  $N^p_k + 2$, where $N^p_k=n_k\cdot b^p_k$. The stages of $P^p_k =
  S^p_{k,1}\cup\{[0:1],[N_k:N_k+1]\}, S^p_{k,2}, \ldots, S^p_{k,D^p_k}$ are
  defined by the following equations:
\begin{align*}
S^p_{k,1} &= \braced{[b^p_ki:b^p_ki+1]: i=1,\ldots,n_k-1}\\ 
S^p_{k,j+s} &= \left\{[b^p_ki+j:b^p_k(i+2^{k-s-1}-1)+(b^p_k-j+1)]: 
                    i=0,\ldots,n_k-2^{k-s-1} \right\}\\ 
& \text{where } 1 \le j \le \frac{b^p_k}{2} \text{ and } (p-2)(j-1) < s 
                               \le \min((p-2)j, k-1);\\ 
S^p_{k,(p-1)j+1} &= \left\{[b^p_ki+j:b^p_ki+j+1],
               [b^p_ki+(b^p_k-j):b^p_ki+(b^p_k-j+1)]: i=0,\ldots,n_k-1 \right\},\\ 
& \text{where } 1 \le j \le \frac{k-2}{p-2}.
\end{align*}
\end{definition}
The networks $P^4_5$ and $P^4_6$ are depicted in Figures \ref{merge62} and 
\ref{merge46}, respectively.
\begin{fact}
$delay(P^p_k) = p$ for any $ k\ge p \ge 4$.  \qed
\end{fact}
Let $A = S_1,S_2,\ldots,S_d$ and $A' = S'_1,S'_2,\ldots,S'_{d'}$ be
$N$-input comparator networks such that for each $i$, $1\le i
\le\min(d,d')$, $regs(S_i) \cap regs(S'_i) = \emptyset$.  Then $A\cup
A'$ is defined to be a network with stages $(S_1 \cup S'_1), (S_2 \cup
S'_2), \ldots, (S_{\max(d,d')} \cup S'_{\max(d,d')})$, where empty
stages are added at the end of the network of smaller depth.

For any comparator network $A=S_1,\ldots,S_d$ and
$D = delay(A)$, let us define a network $B=T_1,\ldots,T_D$ to be a {\em 
compact form} of $A$, where $T_q=\bigcup\braced{S_{q+pD}: 0\le p\le(d-q)/D}$, 
$1\le q\le D$. Observe that $B$ is correctly defined due to the delay
of $A$.  Moreover, $depth(B)=delay(B)=delay(A)$.
\begin{definition}
For $k\ge p\ge 4$ let $M^p_k$ denote the compact form of $P^p_k$ with the first
and the last registers deleted. That is, the network $M^p_k = T^{p,k}_1,\ldots,
T^{p,k}_p$ is using the set of registers numbered $\{1,2, \ldots, N^p_k\}$, where
$N^p_k=n_k\cdot b^p_k$, $n_k = 2^{k-1}-1$, $b^p_k = 2\lceil\frac{k-2}{p-2}\rceil$,
and for $j=1, \ldots, p$ the stage $T^{p,k}_j$ is defined as
$\bigcup\{S^p_{k,j+pi}: 0 \le i \le \frac{D^p_k-j}{p}\}$, where $D^p_k = k-1 + 
\frac{b^p_k}{2}$.
\end{definition}
It is not necessary to delete the first and the last registers of $P^p_k$ but 
this will simplify proofs a little bit in the next section. The network 
$M^4_5$ is given in Fig. \ref{merge3p}.
\begin{figure}
\centering
\epsfig{file=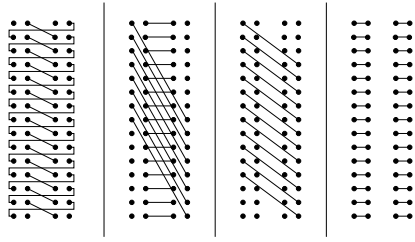, height=1.23in}
\caption{The $M^4_5$ network that is $4$-periodic.}
\label{merge3p}
\end{figure}
\begin{theorem} \label{3merger}\label{3MERGER} 
  For any $p\ge 4$ there exists a family of $p$-periodic comparator networks
  $M^p_k$, $k\ge p$, such that each $M^p_k$ is a $p$-periodic, $(b^p_k-1)$-pass
  merger of two sorted sequences given in odd and even registers, respectively.
  The running time of $M^p_k$ is $p(b_k-1) \le \frac{2p}{p-2}k +
  p\frac{p-8}{p-2} \le \frac{2p}{p-2} \log N^p_k + p\frac{p-8}{p-2}$, where $b^p_k
  = 2\lceil\frac{k-2}{p-2}\rceil$ and $N^p_k = (2^{k-1}-1)\cdot b^p_k$ is the
  number of registers in $M^p_k$.
\end{theorem}
This is the main theorem of the paper. The rest of paper is devoted to its
proof, which is based on the general observation that $M^p_k$ merges
$\lceil\frac{k-2}{p-2}\rceil$ pairs of sorted subsequences, one after another,
in pipeline fashion. Details are given in the next section.

\section{Proof of Theorem \ref{3merger}}\label{sec3} 

The first observation we would like to make is that we can consider
inputs consisting only of 0's and 1's. The well-known Zero-One Principle
states that any comparator network that sorts 0-1 input sequences
correctly sorts also arbitrary input sequences \cite{k}.  In the similar
way, one can prove that the same property holds also for merging:
\begin{proposition}
If a comparator network merges any two 0-1 sorted sequences, then it
correctly merges any two sorted sequences. \qed 
\end{proposition} 
It follows that we can analyze computations of the network $M^p_k$, $k\ge p\ge 4$,
by describing each state of registers as a 0-1 sequence $\overline{x}=(x_1,
\ldots, x_{N^p_k})$, where $x_i$ represents the content of register $i$. If
$\overline{x}$ is an input sequence for $b^p_k-1$ passes
of $M_k$, then by $\overline{x}^{(i)}$ we denote the content of registers after
$i$ passes of $M^p_k$, $i=0, \ldots, b^p_k-1$,, that is,
$\overline{x}^{(0)} = \overline{x}$ and $\overline{x}^{(i+1)} =
M^p_k(\overline{x}^{(i)})$. Since $M^p_k$ consists of $p$ stages $T^{p,k}_1$,
\ldots, $T^{p,k}_p$, we extend the notation to describe the output of each stage:
$\overline{x}^{(i,0)} = \overline{x}^{(i)}$ and $\overline{x}^{(i,j)} =
T^{p,k}_j(\overline{x}^{(i,j-1)})$, for $j=1, \ldots, p$. For other values of $j$
we assume that $\overline{x}^{(i,j)} = \overline{x}^{(i+j\div p,j\bmod p)}$. We
will use this superscript notation for other equivalent representations of
sequence $\overline{x}$.

Now let us fix some technical notations and definitions. A 0-1
sequence can be represented as a word over $\Sigma=\{0,1\}$. A
non-decreasing (also called {\em sorted}) 0-1 sequence has a form of
$0^*1^*$ and can be equivalently represented by the number of ones (or
zeros) in it. For any $x\in\Sigma^*$ let $ones(x)$ denote the number
of $1$ in $x$.  If $x\in\Sigma^n$ then $x_i$, $1\le i\le n$, denotes
the $i$-th letter of $x$. Generally, for $A=\{i_1,\ldots,i_m\}, 1\le
i_1<\ldots,<i_m\le n$, let $x_A$ denotes the word $x_{i_1}\ldots
x_{i_m}$. We say that a sequence $\overline{x} = (x_1, \ldots,
x_{N_k})$ is {\em 2-sorted} if both $(x_1, x_3, \ldots, x_{N_k-1})$
and $(x_2, x_4, \ldots, x_{N_k})$ are sorted.

The roadmap of the proof in the next three subsections is as follows:
\begin{enumerate}
\item In Subsection \ref{ss:1} we reduce the analysis of periodic applications of
our $p$ stages to a 0-1 input to an analysis of periodic applications of $p$ quite
simple functions to a short sequence of integers representing the numbers of ones
in columns. 
\item In Subsection \ref{ss:2} we start the analysis of computations on sequences
with, so called, {\em balanced} sequences. A sequence $(c_1, c_2, \ldots, c_{2n})$
is called balanced if $c_1+c_{2n} = c_2 + c_{2n-1} = \ldots = c_n + c_{n+1}$.
Being balanced is preserved by the simple functions.
\item In subsection \ref{ss:3} we use balanced sequences as upper and lower bounds
on unbalanced sequences. The analysed functions are monotone.
\end{enumerate}

\subsection{Reduction to Analysis of Columns}\label{ss:1} 

For any $k\ge p\ge 4$ let $n_k=2^{k-1}-1$, $b^p_k=2\lceil\frac{k-2}{p-2}\rceil$
(thus $N^p_k=n_k\cdot b^p_k$) and $D^p_k = k-1 + \frac{b^p_k}{2}$. The set of
registers $Reg^p_k=\{1,\ldots,N^p_k\}$ can be analysed as an $n_k\times b^p_k$
matrix with $C^{p,k}_j=\{j+ib^p_k: 0\le i<n_k\}$, $j=1,\ldots,b^p_k$, as columns.
A content of all registers in the matrix, that is $x\in\Sigma^{N_k}$, can be
equivalently represented by the sequence of contents of registers in $C^{p,k}_1$,
$C^{p,k}_2$, \ldots, $C^{p,k}_{b^p_k}$, that is $(x_{C^{p,k}_1}, \ldots,
x_{C^{p,k}_{b^p_k}})$.  Since $b^p_k$ is an even number, the following fact is
obviously true. %
\begin{fact} If $x\in\Sigma^{N^p_k}$ is 2-sorted then each $x_{C^{p,k}_j}$,
  $j=1,\ldots,b^p_k$, is sorted. \qed
\end{fact}
That is, the columns are sorted at the beginning of a computation of
$b^p_k-1$ passes of $M^p_k$. The first lemma we would
like to prove is that columns remain sorted after each stage of the computation.
We start with a following technical fact:
\begin{fact}\label{f33}
Let $A=\{a_1,\ldots,a_n\}$ and $B=\{b_1,\ldots,b_n\}$ be subsets of
$\{1,\ldots,N^p_k\}$ such that $a_1 < b_1 < a_2 < b_2 < \ldots < a_n <
b_n$. Let $h\ge 0$ and $S_{A,B,h}=\{[a_i:b_{i+h}]: 1\le i \le
n-h\}$. Then for any $x\in\Sigma^{N^p_k}$ such that $x_A$ and $x_B$
are sorted, the output $y=S_{A,B,h}(x)$ has the following properties:
\begin{enumerate}[(i)]
\item $y_A$ and $y_B$ are sorted.
\item Let $m_1=ones(x_A)$ and $m_2=ones(x_B)$. Then $ones(y_A) =
  \min(m_1,m_2+h)$ and $ones(y_B) = \max(m_1-h,m_2)$.
\end{enumerate}
\end{fact}
\begin{proof}
To prove (i) we show only that $y_{a_i}\le
y_{a_{i+1}}$ for $i=1,\ldots,n-1$. If $1\le i<n-h$ then $y_{a_i} =
\min(x_{a_i},x_{b_{i+h}}) \le \min(x_{a_{i+1}},x_{b_{i+h+1}}) =
y_{a_{i+1}}$ since $\min$ is a non-decreasing function and both $x_A$
and $x_B$ are sorted . If $i=n-h$ then $y_{a_i} =
\min(x_{a_i},x_{b_{i+h}}) \le x_{a_{i+1}} = y_{a_{i+1}}$. For $i>n-h$
we have $y_{a_i} = x_{a_i} \le x_{a_{i+1}} = y_{a_{i+1}}$.

To prove (ii) let $m'_1=\min(m_1,m_2+h)$ and $m'_2=\max(m_1-h,m_2)$.
We consider two cases. If $m_1\le m_2+h$ then $m_1-h\le m_2$ and we
get $m'_1=m_1$ and $m'_2=m_2$. In this case no comparator from
$S_{A,B,h}$ exchanges 0 with 1. To see this assume a.c. that a
comparator $[a_i:b_{i+h}]$ exchanges $x_{a_i}=1$ with
$x_{b_{i+h}}=0$. Then $i>n-m_1$ and $i+h\le n-m_2$ hold because of the
definitions of $m_1$ and $m_2$. It follows that $n-m_1 < n-m_2-h$,
thus $m_1-h > m_2$ --- a contradiction. If $m_1 > m_2+h$ then $m'_1 =
m_2+h$ and $m'_2=m_1-h$. In this case let us observe that a comparator
$[a_i:b_{i+h}]$ exchanges $x_{a_i}=1$ with $x_{b_{i+h}}=0$ if and only
if $m_2+h \le n-i < m_1$. Therefore $ones(y_A) = m_1-(m_1-m_2-h) =
m_2+h$ and $ones(y_B) = m_2+(m_1-m_2-h) = m_1-h$. \qed
\end{proof}

Since now on we continue the proof for a fixed value $p\ge 4$ and omit $p$ in 
superscripts/subscripts of our denotations, for example, we write $M_k$ instead of 
$M^p_k$.

According to the definition of $M_k$, it consists of stages $T^k_1, \ldots, T^k_p$,
where $T^k_i = \bigcup\{S_{k,i+pj}: 0 \le j\le \frac{D_k-i}{p}\}$ (sets $S_{k,j}$ 
are defined in Def.~\ref{defMk}). Using the notation from Fact \ref{f33}, the
following fact is an easy consequence of Definition \ref{defMk}.
\begin{fact}\label{f34} Let $L_i=C^k_i$ and $R_i=C^k_{b_k-i+1}$ denote the
  corresponding left and the right columns of registers, and
  $h_i=2^{k-i-1}-1$, $i=1,\ldots,\frac{b_k}{2}$. Then
\begin{enumerate}[(i)]
\item $regs(S_{k,1})\subseteq L_1\cup R_1$ and $S_{k,1} =
  S_{R_1-\{N_k\},L_1-\{1\},0}$;
\item $regs(S_{k,j+s})\subseteq L_j\cup R_j$ and
  $S_{k,j+s} = S_{L_j,R_j,h_s}$, for $j=1,\ldots,\frac{b_k}{2}$ and $(p-2)(j-1) < 
  s \le \min((p-2)j, k-1)$;
\item $regs(S_{k,(p-1)j+1})\subseteq L_j\cup L_{j+1}\cup R_{j+1}\cup
  R_j$ and $S_{(p-1)j+1} = S_{L_j,L_{j+1},0}\cup S_{R_{j+1},R_j,0}$, for
  $j=1,\ldots,\frac{b_k}{2}-1$;
\item $regs(S_{k,D_k})\subseteq L_{b_k/2}\cup R_{b_k/2}$ and
  $S_{k,D_k} = S_{L_{b_k/2},R_{b_k/2},0}$;
\item if $(L_j\cup R_j)\cap regs(S_{k,i})\neq\emptyset$ then $(p-1)(j-1)+1\le
  i\le \min((p-1)j+1,D_k)$, for any $j=1,\ldots,\frac{b_k}{2}$.  \qed
\end{enumerate} 
\end{fact}
\begin{lemma} \label{l35} If the initial content of registers is a
  2-sorted 0-1 sequence $x$ then after each stage of multi-pass
  computation of $M_k=T^k_1, \ldots,T^k_p$ the content of each column $C^k_j$,
  $j=1,\ldots,b_k$, is sorted, that is, each $(x^{(s,i)})_{C^k_j}$ is of
  the form $0^*1^*$, $s=0,\ldots$, $i=1, \ldots, p$.
\end{lemma}
\begin{proof} By induction it suffices to prove that for each sequence
  $y\in\Sigma^{N_k}$ with sorted columns $C^k_j$, $j=1,\ldots,b_k$, the
  outputs $z_i=T^k_i(y)$, $i=1, \ldots, p$ have also the columns sorted. Since
  each $T^k_i$, as a mapping, is a composition of mapping $S_{k,i+pj}, 0 \le
  j\le \lfloor \frac{D_k-i}{p} \rfloor$, each of which, due to Facts 
  \ref{f33} and \ref{f34}, transforms sorted columns into sorted columns, the
  lemma follows.  \qed
\end{proof}

From now on, instead of looking at 0-1 sequences with sorted columns, we
will analyse the computations of $M_k$ on sequences of integers
$\overline{c}=(c_1,\ldots,c_{b_k})$, where $c_t$, $t=1,\ldots,b_k$,
denote the number of ones in a sorted column $C^k_t$. Transformations of
0-1 sequences defined by sets $S_{k,j}$, $j=1, \ldots, D_k$ will be
represented by the following mappings:
\begin{definition} \label{defFun}
Let $k\ge p$, $h_i=2^{k-i-1}-1$ for $i = 1, \ldots, k-1$ and $b_k =
2\lceil\frac{k-2}{p-2}\rceil$. For $j = 1, \ldots, \frac{b_k}{2}$ and  $s = 1, 
\ldots, k-1$ the functions $dec^k_{j,s}$, $mov^k_j$ and $cyc^k$ over
sequences of $b_k$ reals are defined as follows. Let $\overline{c}=(c_1, \ldots,
c_{b_k})$ and $t\in\{1,\ldots,b_k\}$.
\begin{align*}
(dec^k_{j,s}(\overline{c}))_t &= \begin{cases}
\min(c_j,c_{b_k-j+1}+h_s) & \text{if $t=j$}\\
\max(c_j-h_s,c_{b_k-j+1}) & \text{if $t=b_k-j+1$}\\
c_t                     & \text{otherwise}\\ 
\end{cases}\\
(mov^k_j(\overline{c}))_t &= \begin{cases}
\min(c_t,c_{t+1}) & \text{if $t=j$ or $t=b_k-j$}\\
\max(c_{t-1},c_t) & \text{if $t=j+1$ or $t=b_k-j+1$}\\
c_t              & \text{otherwise}\\ 
\end{cases}\\
(cyc^k(\overline{c}))_t &= \begin{cases}
\max(c_1,c_{b_k}-1) & \text{if $t=1$}\\
\min(c_1+1,c_{b_k}) & \text{if $t=b_k$}\\
c_t                & \text{otherwise}\\ 
\end{cases}
\end{align*}
\end{definition}
\begin{fact} \label{fct-5}
Let $x\in\Sigma^{N_k}$ be a 0-1 sequence with sorted columns $C^k_1, \ldots,
  C^k_{b_k}$, let $c_i=ones(x_{C^k_i})$ and $\overline{c} =
  (c_1,\ldots,c_{b_k})$. Let $y_j=S_{k,j}(x)$, $d_{j,i}=ones((y_j)_{C^k_i})$
  and $\overline{d_j} =(d_{j,1},\ldots,d_{j,b_k})$, where
  $i = 1, \ldots, b_k$ and $j = 1 ,\ldots, D_k$. Then
\begin{enumerate}[(i)]
\item $\overline{d_1} = cyc^k(\overline{c})$
\item $\overline{d_{j+s}} = dec^k_{j,s}(\overline{c})$, for any
  $j=1,\ldots,\frac{b_k}{2}$ and $(p-2)(j-1) < s \le \min((p-2)j, k-1)$
\item $\overline{d_{(p-1)j+1}} = mov^k_j(\overline{c})$, for any
  $1 \le j \le \frac{k-2}{p-2}$
\end{enumerate}
\end{fact}
\begin{proof} Generally, the fact follows from Fact \ref{f34} and the
  part (ii) of Fact \ref{f33} We prove only its parts (i) and (ii). Part
  (iii) can be proved in a similar way.

\textit{(i)~} Observe that $y_1=S_{k,1}(x)=S_{R_1-\{N_k\},L_1-\{1\},0}(x)$
  due to Fact \ref{f34} \textit{(ii)}. It follows that only the content
  of columns $L_1=C^k_1$ and $R_1=C^k_{b_k}$ can change, but they remain
  sorted (according to Lemma \ref{l35}). Using Fact
  \ref{f33} \textit{(ii)} we have: $m_1 = ones(x_{R_1-\{N_k\}}) =
  c_{b_k}-x_{N_k}$, $m_2 = ones(x_{L_1-\{1\}}) = c_1-x_1$ and
\[d_{1,1} = \max(m_1,m_2)+x_1 = \max(c_{b_k}-x_{N_k}+x_1,c_1),\]
\[d_{1,b_k} = \min(m_1,m_2)+x_{N_k} = \min(c_{b_k},c_1+x_{N_k}-x_1).\]
Now let us consider the following three cases of values $x_1$ and
$x_{N_k}$: \\[3pt]
\noindent\textbf{Case $x_1=0$ ~and~ $x_{N_k}=1$.} Then 
  $d_{1,1} = \max(c_{b_k}-1,c_1) = cyc^k(\overline{c})_1$ and 
  $d_{1,b_k} = \min(c_{b_k},c_1+1)=cyc^k(\overline{c})_{b_k}$.\\[3pt]
\noindent\textbf{Case $x_1=1$.} Then $c_1=n_k$, $c_{b_k}\le n_k$ and
  $c_{b_k}-x_{N_k}\le n_k-1$. In this case: $d_{1,1} = \max(n_k,c_{b_k}-x_{N_k}+1)
  = n_k = \max(c_1,c_{b_k}-1) = cyc^k(\overline{c})_1$ and 
  $d_{1,b_k} = \min(n_k-1+x_{N_k},c_{b_k}) = c_{b_k} = \min(c_1+1,c_{b_k})
  = cyc^k(\overline{c})_{b_k}$.\\[3pt]
\noindent\textbf{Case $x_{N_k}=0$.} Then $c_{b_k}=0$ and $c_1-x_1\ge
  0$. In this case: $d_{1,1} = \max(c_1,x_1) = c_1 = \max(c_1,c_{b_k}-1) = 
  cyc^k(\overline{c})_1$ and $d_{1,b_k} = \min(c_1-x_1,c_{b_k}) = c_{b_k}
  = \min(c_1+1,c_{b_k}) = cyc^k(\overline{c})_{b_k}$.

\textit{(ii)~} We fix any $j\in\{1,\ldots,\frac{b_k}{2}\}$ and $(p-2)(j-1) < s \le
\min((p-2)j, k-1)$ and observe that $y_{j+s}=S_{k,j+s}(x)=S_{L_j,R_j,h_s}(x)$ due
to Fact \ref{f34} \textit{(ii)}. It follows that only the content of columns
$L_j=C^k_j$ and $R_j=C^k_{b_k-j+1}$ can change, but they remain sorted (according
to Lemma \ref{l35}). Using Fact \ref{f33} \textit{(ii)} we have:
\begin{align*}
d_{j+s,j} = ones((y_{j+s})_{L_j}) = \min(c_j,c_{b_k-j+1}+h_s) &= 
  (dec^k_{j,s}(\overline{c}))_j,\\
d_{j+s,b_k-j+1} = ones((y_{j+s})_{R_j}) = \max(c_j-h_s,c_{b_k-j+1}) &= 
  (dec^k_{j,s}(\overline{c}))_{b_k-j+1}.
\end{align*} \qed
\end{proof}
\begin{definition} \label{defQ}
Let $k\ge p$. For $x = 1, \ldots, k$ let $MV^k_x = \{mov^k_j: 1\le j\le 
\frac{k-2}{p-2} \text{ and } x+j \equiv 1 \pmod p\}$ and $DC^k_x =  
\{dec^k_{j,s}: 1\le j\le \frac{b_k}{2} \text{ and } (x+j)\bmod p \notin \{1,2\} 
\text{ and } s = (p-2)(j-1) - 1 + (x + j - 1) \bmod p \le k-1\}$. Let 
$Q^k_1$, \ldots, $Q^k_p$ denote the following sets of functions.
\begin{align*}
Q^k_1 &= \left\{cyc^k\right\} \cup MV^k_1 \cup DC^k_1\\
Q^k_i &=   MV^k_i \cup DC^k_i \quad\text{ for } 2 \le i \le p.
\end{align*}
\end{definition}
Let us observe that each function in $Q^k_i$, $i=1, \ldots, p$, can modify only
a few positions in a given sequence of numbers. Moreover, different
functions in $Q^k_i$ modify disjoint sets of positions. For a
function $f:R^m\mapsto R^m$ let us
define \[args(f)=\left\{i\in\{1,\ldots,m\}:\exists_{\overline{c}\in R^m}
(f(\overline{c}))_i\neq (\overline{c})_i\right\}\] The following facts
formalize our observations.
\begin{fact}\label{fct-6}
Let $k\ge p$. Then $args(cyc^k)=\{1,b_k\}$, $args(dec^k_{j,s})=\{j,b_k-j+1\}$,
$args(mov^k_j)=\{j,j+1,b_k-j,b_k-j+1\}$, where $j=1,\ldots,\frac{b_k}{2}$.
\end{fact} 
\begin{fact}\label{fct-7}
For each pair of functions $f,g\in Q^k_i$, $f\neq g$, $i=1, \ldots, p$, we have
\begin{enumerate}[(i)]
\item $args(f) \cap args(g) = \emptyset$;
\item for any $\overline{c}=(c_1,\ldots,c_{b_k})$ and 
                                                  $j\in\{1,\ldots,b_k\}$
\[
(f(g(\overline{c})))_j =  \begin{cases}
(f(\overline{c}))_j & \text{if $j\in args(f)$}\\
(g(\overline{c}))_j & \text{if $j\in args(g)$}\\
c_j                 & \text{otherwise}\\
\end{cases}
\]
\end{enumerate} 
\end{fact}
\begin{proof}
(i) Assume a.c. that there exist $1 \le x \le p$, $f,g \in Q^k_x$ and  $1 \le j \le 
b_k/2$ such that $f \neq g$ and $j \in args(f) \cap args(g)$. Obviously, functions 
$f$ and $g$ cannot be both in $MV^k_x$ or $DC^k_x$. Assume that $f \in MV^k_x \cup 
\{cyc^k\}$ and $g \in DC^k_x$. Then $(x+j) \bmod p \in \{1, 2\}$ from the first 
assumption and $(x+j) \bmod p \notin \{1, 2\}$ from the second one - a 
contradiction. \qed
\end{proof}
\begin{corollary} \label{color1}
Each set $Q^k_i$, $i=1, \ldots, p$, uniquely determines a mapping, in which
functions from $Q^k_i$ can be apply in any order. Moreover, if $f\in
Q^k_i$, $\overline{c}\in R^{b_k}$ and $j\in args(f)$ then
$(Q^k_i(\overline{c}))_j = (f(\overline{c}))_j$. \qed
\end{corollary}
We would like to prove that the result of applying each $Q^k_i$, $i=1, \ldots, p$,
to a sequence $\overline{c}=(c_1,\ldots,c_{b_k})$ of numbers of ones in
columns $C^k_1, \ldots, C^k_{b_k}$ is equivalent to applying the set of
comparators $T^k_i$ to the content of registers, if each column is
sorted.
\begin{lemma} \label{l2}
Let $x\in\Sigma^{N_k}$ be a 0-1 sequence with sorted columns $C^k_1, \ldots,
  C^k_{b_k}$, let $c_i=ones(x_{C^k_i})$ and $\overline{c} =
  (c_1,\ldots,c_{b_k})$. Let $y_j=T^k_j(x)$, $d_{j,i}=ones((y_j)_{C^k_i})$
  and $\overline{d_j} =(d_{j,1},\ldots,d_{j,b_k})$, where
  $i=1,\ldots,b_k$ and $j=1, \ldots, p$. Then $Q^k_j(\overline{c}) = 
  \overline{d_j}$.
\end{lemma}
\begin{proof}
Recall that $T^k_j=\bigcup\{S_{k,j+pi}: 0\le i\le
\frac{D_k-j}{p}\}$. For a set of comparators $S$ let us define
\[ cols(S) = \left\{ i\in\{1,\ldots,b_k\}: regs(S)\cap C^k_i\neq\emptyset
\right\} \enspace .\] From Fact \ref{f34}(i--iv) it follows that
$cols(S_{k,1})=\{1,b_k\}$ and for $j=1,\ldots,\frac{b_k}{2}$ $cols(S_{k,j+s}) =
\{j,b_k-j+1\}$ and $ cols(S_{k,(p-1)j+1}) = \{j, j+1, b_k-j, b_k-j+1\}
$. From Fact \ref{f34}(v) we get that $ cols(S_{k,j+pi}) \cap
cols(S_{k,j+pi'}) = \emptyset $ if $i\neq i'$. Thus we can observe a 1-1
correspondence between a function $f$ in $Q^k_j$ and a set of
comparators $S_{k,j+pi} \subseteq T^k_j$ such that
$args(f)=cols(S_{k,j+pi})$ Then for each $t\in args(f)$
$(Q^k_j(\overline{c}))_t = (f(\overline{c}))_t = (\overline{d_j})_t$, as
the consequence of Corollary \ref{color1} and Fact \ref{fct-5}. \qed
\end{proof}
\begin{definition} \label{flat}
We say that a sequence of numbers $\overline{c} = (c_1,\ldots,c_{2m})$
is {\em flat} if $c_1\le c_2\le \ldots, c_{2m}\le c_1 + 1$. We say that
a sequence $\overline{c}$ is {\em 2-flat} if subsequences
$(c_1,c_3,\ldots,c_{2m-1})$ and $(c_2,c_4,\ldots,c_{2m})$ are flat. We
say that $\overline{c}$ is balanced if $c_i+c_{2m-i+1} = c_1+c_{2m}$,
for $i=2,\ldots,m$. For a balanced sequence $\overline{c}$ define 
$height(\overline{c})$ as $c_1+c_{2m}$.
\end{definition}
\begin{proposition} \label{p3}
Let $k\ge p$, $x\in \Sigma^{N_k}$, $\overline{c} =
(c_1,\ldots,c_{b_k})$, where $c_i=ones(x_{C^k_i})$ ($C^k_i$ is as usual a
column in the matrix of registers), $i=1,\ldots,b_k$. Then
\begin{enumerate}[(i)]
\item $x$ is sorted if and only if columns of $x$ are sorted  and 
$\overline{c}$ is flat;
\item $x$ is 2-sorted if and only if columns of $x$ are sorted and 
$\overline{c}$ is 2-flat;
\end{enumerate}
\end{proposition} 
Now we are ready to reduce the proof of Theorem \ref{3merger} to the
proof of following lemma.
\begin{lemma} \label{l3}
Let $k\ge p\ge 4$ and $b_k = 2\lceil\frac{k-2}{p-2}\rceil$. If for each 2-flat
sequence $\overline{c} = (c_1,\ldots,c_{b_k})$ of integers from $[0,2^{k-1}-1]$
the result of application $(Q^k_p\circ \ldots\circ Q^k_1)^{b_k-1}$ to
$(\overline{c})$ is a flat sequence, then $M_k$ is a $(b_k-1)$-pass merger of two
sorted sequences given in odd and even registers, respectively.
\end{lemma}
\begin{proof}
Assume that for each 2-flat sequence $\overline{c} = (c_1,\ldots,c_{b_k})$ the
result of application $(Q^k_p\circ \ldots\circ Q^k_1)^{b_k-1}$ to $(\overline{c})$
is a flat sequence. Let $\overline{x}\in \Sigma^{N_k}$ be a 2-sorted sequence and
$\overline{c} = (c_1,\ldots,c_{b_k})$, where $c_i=ones(\overline{x}_{C^k_i})$
($C^k_i$ is as usual a column in the matrix of registers), $i=1,\ldots,b_k$. Then
$\overline{c}$ is 2-flat due to Proposition \ref{p3} and each
$c_i\in[0,2^{k-1}-1]$, because the height of columns is $2^{k-1}-1$. Recall that
$\overline{x}^{(j)}=(M_k)^j(\overline{x})$ and let
$c_{j,i}=ones(\overline{x}^{(j)}_{C^k_i})$. Using Lemma \ref{l2} and easy
induction we get that the equality $(Q^k_p\circ \ldots\circ
Q^k_1)^{j}(\overline{c}) = (c_{j,1},\ldots,c_{j,b_k})$ is true for $j=1,\ldots,
b_k-1$. Since the result of $(Q^k_p\circ \ldots\circ Q^k_1)^{b_k-1}(\overline{c})$
is a flat sequence, the sequence $\overline{x}^{(b_k-1)}$ is sorted. \qed
\end{proof}

\subsection{Analysis of Balanced Columns}\label{ss:2} 

Due to Lemma \ref{l3} we can analyse only the results of periodic
application of the functions $Q^k_1, \ldots, Q^k_p$ to a
sequence of integers representing the numbers of ones in each register
column. We know also that an initial sequence is 2-flat. To simplify
our analysis further, we start it with initial values restricted to be
balanced 2-flat sequences. In this section we prove that after $p(b_k-1) - b_k/2+1$
such application to a balanced 2-flat sequence we get a flat output
sequence (see Lemma \ref{merge-bal}). Then we observe that the
functions are monotone and any 2-flat sequence can be bounded from
below and above by balanced 2-flat sequences whose heights differ at most
by one. Using these facts we analyse general 2-flat sequences in the
next section.
\begin{lemma} \label{l4}
Let $k\ge p$ and $\overline{c} = (c_1,\ldots,c_{b_k})$ be a balanced
sequence of numbers. Let $s = height(\overline{c})$ and let $f$ be a
function from $Q^k_1\cup \ldots \cup Q^k_p$. Then $f(\overline{c})$ is
also balanced and $height(f(\overline{c}))=s$.
\end{lemma}
\begin{proof}
Let $\overline{c}$ and $s$ be as assumed in the lemma and let $f(\overline{c}) =
(d_1,\ldots,d_{b_k})$.  The function $f\in Q^k_1\cup \ldots \cup Q^k_p$ can be
either $cyc^k$ or one of $mov^k_j$, $dec^k_{j,t}$, where $j=1, \ldots, b_k/2$ and
$t = 1, \ldots, k-1$, according to Definition \ref{defQ}. Each of the functions
can modify only one or two pairs of positions of the form $(i,b_k-i+1)$ in
$\overline{c}$ (see Fact \ref{fct-6}). The other pairs are left untouched, so the
sum of their values cannot change. In case of $cyc^k$ the modified pair is
$(1,b_k)$ and $d_1+d_{b_k} = \max(c_1,c_{b_k}-1) + \min(c_1+1,c_{b_k}) = s$. In
case of $dec^k_{j,t}$ the pair is $(j,b_k-j+1)$ and $ d_j+d_{b_k-j+1} =
\min(c_j,c_{b_k-j+1}+h_t) + \max(c_j-h_t,c_{b_k-j+1}) = \min(c_j-h_t,c_{b_k-j+1})
+ h_t + \max(c_j-h_t,c_{b_k-j+1}) = s$.  Finally, if $f=mov^k_j$ then we have two
pairs $(j,b_k-j+1)$ and $(j+1,b_k-j)$. Then $ d_j+d_{b_k-j+1} = \min(c_j,c_{j+1})
+ \max(c_{b_k-j},c_{b_k-j+1}) = \min(c_j,c_{j+1}) + \max(s-c_{j+1},s-c_j) = s$. In
case of the other pair $d_{j+1}+d_{b_k-j} = \max(c_j,c_{j+1}) + \min(c_{b_k-j},
c_{b_k-j+1}) = \max(c_j,c_{j+1}) + \min(s-c_{j+1},s-c_j) = s$. \qed
\end{proof}
It follows from Lemma \ref{l4} that if we start the periodic
application of the functions $Q^k_1$, \ldots, $Q^k_p$ to a balanced
2-flat initial sequence then it remains balanced after each function
application and its height will not changed. Therefore, we can trace 
only the values in the first half of generated sequences. If needed, a
value in the second half can be computed from the height and the
corresponding value in the first half. To get a better view on the
structure of generated sequences, we subtract half of the height from
each element of the initial sequence and proceed with such modified
sequences to the end. At the end the subtracted value is added to each
element of the final sequence. The following fact justifies the described
above procedure.
\begin{fact}\label{fct-10}
Let $f$ be a function from $Q^k_1\cup \ldots \cup Q^k_p$. Then $f$ is monotone and
for each $t\in R$ and $(c_1,\ldots,c_{b_k}) \in R^{b_k}$ the following equation is
true
\[ f(c_1-t,\ldots,c_{b_k}-t) =  f(c_1,\ldots,c_{b_k}) - (t,\ldots,t) \enspace .\]
\end{fact}
\begin{proof}
The fact follows from the similar properties of $\min$ and $\max$
functions: they are monotone and the equations: $\min(x-t,y-t) =
\min(x,y)-t$ and $\max(x-t,y-t) = \max(x,y)-t$ are obviously true. Each
$f$ in $Q^k_1\cup \ldots \cup Q^k_p$ is defined with the help of these
simple functions, thus $f$ inherits the properties. \qed
\end{proof}
\begin{corollary}
Let $f = f_l\circ f_{l-1}\circ\ldots\circ f_1$, where each $f_i$ is
from $\{Q^k_1, \ldots, Q^k_p\}$, $1\le i\le l$. Then $f$ is monotone
and for any $t\in R$ and $(c_1,\ldots,c_{b_k})\in R^{b_k}$
\[ f(c_1-t,\ldots,c_{b_k}-t) =  
             f(c_1,\ldots,c_{b_k}) - (t,\ldots,t) \enspace .\] 
\end{corollary} 
\begin{definition} \label{reduce}
Let $\overline{c}=(c_1,\ldots,c_{b_k})\in R^{b_k}$ be a balanced
sequence and $s=height(\overline{c})$. We call $(c_1-\frac{s}{2},
c_2-\frac{s}{2}, \ldots, c_{b_k/2}-\frac{s}{2})\in R^{b_k/2}$ the reduced
sequence of $\overline{c}$ and denote it by $reduce(\overline{c})$. For
a sequence $\overline{d} = (d_1, \ldots, d_{b_k/2})\in R^{b_k/2}$ we define
$s$-extended sequence $ext(\overline{d},s)$ as
\[(d_1+\frac{s}{2}, d_2+\frac{s}{2}, \ldots, d_{b_k/2}+\frac{s}{2}, 
  \frac{s}{2}-d_{b_k/2}, \frac{s}{2}-d_{k-3}, \ldots, \frac{s}{2}-d_1) \enspace .\]

For any $t\in R$ and a function $f: R^{b_k}\mapsto R^{b_k}$ that maps
balanced sequences to balanced ones and preserves heights let
$reduce(f,t)$ denote a function on $R^{b_k/2}$ such that for any
$\overline{d}\in R^{b_k/2}$
\[(reduce(f,t))(\overline{d}) = reduce(f(ext(\overline{d},t)))\].
\end{definition}
Observe that for a balanced sequence $\overline{c}$ with height $s$ the
sequence $ext(reduce(\overline{c}),s)$ is equal to
$\overline{c}$. Moreover, for any $t\in R$ and a sequence
$\overline{d}\in R^{b_k/2}$ the sequence $ext(\overline{d},t)$ is balanced
and its height is $t$, thus $reduce(ext(\overline{d},t)) =
\overline{d}$. Note also that functions $Q^k_1$, \ldots, $Q^k_p$
preserve the property of being balanced and the sequence height (see
Lemma \ref{l4}), so we can analyse a periodic application of their
reduced forms to a reduced balanced 2-flat input.
\begin{fact}
Let $f = f_l\circ f_{l-1}\circ\ldots\circ f_1$, where $f_i\in\{Q^k_1,
\ldots, Q^k_p\}$, $1\le i\le l$. Let $\overline{c}\in R^{b_k}$ be
balanced and $s=height(\overline{c})$ Let $\hat{f_i} = reduce(f_i,s)$,
$1\le i\le l$, and $\hat{f} = \hat{f_l} \circ \hat{f_{l-1}} \circ
\ldots \circ \hat{f_1}$. Then $f(\overline{c}) =
ext((\hat{f})(reduce(\overline{c})),s)$.
\end{fact}
\begin{definition}
Define $MinMax(x,y)$ to be $(min(x,y),\max(x,y))$, $Min(x) =
\min(x,-x)$ and $Cyc(x) = \max(x,-x-1)$. Let $Dec_i(x) = \min(x,-x+H_i)$,
where $H_i = 2^i-1, i=1,\ldots$.
\end{definition}
\begin{fact} \label{fct-reduce}
Let $k \ge p$. For each $f \in Q^k_1 \cup \ldots \cup Q^k_p$ and $t \ge 0$ the
function $reduce(f,t)$ does not depend on $t$ and for any sequence $\overline{d}
\in R^{b_k/2}$ and an index $u$, $1\le u \le \frac{b_k}{2}$ the following
equations are true:
\begin{enumerate}[(i)]
\item $(reduce(cyc^k,t)(\overline{d}))_u = \text{if }u = 1 \text{ then } Cyc(d_1) 
\text{ else } d_u$;
\item $(reduce(dec^k_{j,s},t)(\overline{d}))_u = \text{if }u = j \text{ then } 
Dec_{k-s-1}(d_j) \text{ else } d_u$;
\item $(reduce(mov^k_j,t)(\overline{d}))_u = \text{if }u \in \{j,j+1\} \text{ then 
} (MinMax(d_j,d_{j+1}))_{u-j+1} \text{ else } d_u$, for $j < b_k/2$;
\item $(reduce(mov^k_{b_k/2},t)(\overline{d}))_u = \text{if }u = b_k/2 \text{ then 
} Min(d_{b_k/2}) \text{ else } d_u$.
\end{enumerate} 
\end{fact}
\begin{proof}
By Lemma \ref{l4} the considered functions preserve the height of sequences and 
their property of being balanced, thus we can used their reduced forms and 
$(reduce(f,t))(\overline{d}) = reduce(f(ext(\overline{d},t)))$.
If $u\notin args(f)$ then $(reduce(f(ext(\overline{d},t))))_u = d_u$ according to 
Def.~\ref{reduce}. If  $u\in args(f)$ then we have to consider the following 
cases.\\[3pt]
\noindent\textbf{Case $f = cyc^k$.} Then $u$ must be equal to $1$ and 
$(reduce(cyc^k(ext(\overline{d},t))))_1 = \max(d_1+\frac{t}{2},\frac{t}{2}-d_1-1) - 
\frac{t}{2} = \max(d_1,-d_1 - 1) = Cyc(d_1)$.\\[3pt]
\noindent\textbf{Case $f = dec^k_{j,s}$.} Then $u$ must be equal to $j$ and 
$(reduce(dec^k_{j,s}(ext(\overline{d},t))))_j = 
\min(d_j+\frac{t}{2},\frac{t}{2}-d_j+h_s) - \frac{t}{2} = \min(d_j,-d_j + h_s) = 
Dec_{k-s-1}(d_1)$, because $h_s = 2^{k-s-1}-1 = H_{k-s-1}$.\\[3pt]
\noindent\textbf{Case $f = mov^k_j$ and $j < \frac{b_k}{2}$.} Then $u \in \{j, 
j+1\}$. For $u = j$, $(reduce(mov^k_j(ext(\overline{d},t))))_j 
= \min(d_j+\frac{t}{2},d_{j+1}+\frac{t}{2}) - \frac{t}{2} = \min(d_j,d_{j+1}) = 
(MinMax(d_j,d_{j+1}))_1$. For $u = j+1$ the proof is similar.\\[3pt]
\noindent\textbf{Case $f = mov^k_j$ and $j = \frac{b_k}{2}$.} Then $u$ must be 
 $b_k/2$ and $(reduce(mov^k_{b_k/2}(ext(\overline{d},t))))_{b_k/2} = 
\min(d_{b_k/2}+\frac{t}{2},\frac{t}{2}-d_{b_k/2}) - \frac{t}{2} = 
\min(d_{b_k/2},-d_{b_k/2}) = Min(d_{b_k/2})$. \qed
\end{proof}
\begin{definition}\label{def-Cyc}
Let $k\ge p$ and for each $f \in Q^k_1 \cup \ldots \cup Q^k_p$ let $\hat{f}$
denote its reduced form $reduce(f,*)$ (it does not depend on the second argument).
Let $\hat{Q}^k_1$, \ldots, $\hat{Q}^k_p$ denote the following sets of reduced
functions: $\hat{Q}^k_i = \{\hat{f}: f \in Q^k_j\}$, where $i=1, \ldots, p$.
\end{definition}
\begin{lemma}\label{l5}
Let $k\ge p$ and $t\in R$. Then the function $reduce(Q^k_i,t)$ does not depend on
$t$ and $reduce(Q^k_i,t) = \hat{Q}^k_i$, where $i=1, \ldots, p$.
\end{lemma}
\begin{proof}
Let $f$ be any function in $Q^k_1 \cup \ldots \cup Q^k_p$ and let $\hat{f}$ denote
its reduced form. By Fact \ref{fct-reduce}, we know that $args(\hat{f}) = args(f)
\cap \{1, \ldots, \frac{b_k}{2}\}$. By the definitions, $args(Q^k_i) =
\bigcup_{f\in Q^k_i} args(f)$ and  $args(\hat{Q}^k_i) = \bigcup_{f\in Q^k_i}
args(\hat{f})$. Consider now a sequence $\overline{d} \in R^{b_k/2}$ and an index
$u$, $1\le u \le \frac{b_k}{2}$. If $u \notin args(Q^k_i)$ then
$(reduce(Q^k_i,t)(\overline{d}))_u = d_u = (\hat{Q^k_i}(\overline{d})_u)$.
Otherwise, if $u \in args(f)$, $f \in Q^k_i$, then
$(reduce(Q^k_i,t)(\overline{d}))_u = (reduce(Q^k_i(ext(\overline{d},t))))_u =
(reduce(f(ext(\overline{d},t))))_u = (\hat{f}(\overline{d}))_u = 
(\hat{Q^k_i}(\overline{d}))_u$.  \qed
\end{proof}

Instead of tracing individual values in reduced sequences after each application
of a function from $\{\hat{Q}^k_1, \ldots, \hat{Q}^k_p\}$ we will trace intervals
in which the values should be and observe how the lengths of intervals are
decreasing during the computation. So let us now define the intervals and give a
fact about computations on them.
\begin{definition} \label{interval}
Let $k\ge 4$, $H_i=2^i-1$ for $1\le i\le k-1$. Let $I(0)$ denote the
interval $[-\frac 12,0]$ and, in similar way, let $I(i) = [-\frac
  12,\frac{H_i}{2}]$, $1\le i\le k-1$, $I(-k) = [-\frac{H_{k-1}}{2},0]$
and $I(\pm k) = [-\frac{H_{k-1}}{2},\frac{H_{k-1}}{2}]$. Moreover, we
write $I(w_1,w_2,\ldots,w_l)$ for the Cartesian product $I(w_1)
\times I(w_2) \times \ldots \times I(w_l)$, where each $w_i \in
\{0,1,2,\ldots,k-1,-k,\pm k\}$.
\end{definition}
\begin{fact} \label{fct-12}
The following inclusions are true:
\begin{enumerate}
\item $Dec_i(I(i+1))\subseteq I(i)$ and 
  $Dec_i(I(w))\subseteq I(w)$, for $1\le i\le k-2$ and $w\in\{0,-k,\pm k\}$;
\item $Cyc(I(-k)) \subseteq I(k-1)$ and 
  $Cyc(w)\subseteq Cyc(w)$, for $w\in\{0,k-1\}$;
\item $Min(I(\pm k)) \subseteq I(-k)$ and $Min(I(1)) \subseteq I(0)$;
\item $MinMax(I(\pm k,-k)) \subseteq (I(-k,\pm k))$;
\item $MinMax(I(i,w)) \subseteq (I(w,i))$, 
                              for $1\le i\le k-1$ and $w\in\{0,-k\}$.
\end{enumerate}
\end{fact}
\begin{proof}
The proof of each inclusion is a straightforward consequence of the
definitions of a given function and intervals. Therefore we check only
inclusions given in the first item. Let $x\in I(i+1) = [-\frac
  12,\frac{H_{i+1}}{2}]$. If $x\in I(i) = [-\frac
  12,\frac{H_i}{2}]$. then $Dec_i(x) = \min(x,-x+H_i) = x$ since $2x\le
H_i$. Otherwise $x$ must be in $(\frac{H_i}{2},\frac{H_{i+1}}{2}]$, but
  then $x > -x+H_i$ and $Dec_i(x) = -x+H_i \in [-\frac
    12,\frac{H_i}{2})$ since $H_{i+1} = 2H_i+1$.

To proof the second inclusion for $Dec_i$ let us observe that if $x\le 0$
then $Dec_i(x) = x$. It follows that $Dec_i(I(0))\subseteq I(0)$ and
$Dec_i(I(-k))\subseteq I(-k)$. In case of $x\in I(\pm k)$ we have
to check only the positive values of $x$. such that $x\ge -x+H_i$. But then
$Dec_i(x) = -x+H_i > -x$ and both $x,-x\in I(\pm k)$. \qed
\end{proof}

Now we are ready to define sequences of intervals that are used to describe states
of computation after each periodic application of functions $\hat{Q}^k_1$, \ldots,
$\hat{Q}^k_p$ to a reduced sequence of numbers of ones in columns.

\begin{definition} \label{def-seq}
Let $k\ge p$. For $0 \le x \le p$ and $1 \le l \le b_k/2$ let 
\[e^k_p(x,l) = \max(0, k - (p-2)(l-1) - (x + l -1 ) \bmod p)\]
be an auxiliary function to define the following sequences of length $b_k/2$
\begin{align*}
(U^k_x)_l &= \text{if~~~} x+l \equiv 1 \pmod p \text{~~then~~} -k \text{~~else~~} 
\pm k,\\ 
(V^k_x)_l &= \text{if~~~} x+l \equiv 1 \pmod p \text{~~then~~} -k \text{~~else~~} 
e^k_p(x,l),\\ 
(W^k_x)_l &= \text{if~~~} x+l \equiv 1 \pmod p \text{~~then~~} 0 \text{~~else~~} 
e^k_p(x,l),\\ 
(Z^k)_l   &= 0.
\end{align*}
\end{definition}

Note that the elements of the defined above sequences are interval descriptors as 
defined in Definition \ref{interval} and we have also $U^k_0 = U^k_p$, $V^k_0 = 
V^k_p$ and $W^k_0 = W^k_p$. 
\begin{definition}
Let $k\ge p$. Let $\overline{a}=(a_1,\ldots,a_n)$ and
$\overline{b}=(b_1,\ldots,b_n)$ be any sequences, where $n\ge \frac{b_k}{2}$. For
$0\le i\le \frac{b_k}{2}$ let $join_k(i,\overline{a},\overline{b})$ denote $(a_1,
\ldots, a_i, b_{i+1}, \ldots, b_{b_k/2})$.
\end{definition}

\begin{definition}\label{def-X}
Let $k\ge p$. Let $X^k_i$ denote a state sequence after $i$ stages and
be defined as:
\[ X^k_i = \begin{cases}
join_k(\lceil\frac{i+1}{p-1}\rceil,V^k_{i\bmod p},U^k_{i\bmod p}) & 1\le i\le 
\frac{b_k}{2}(p-1)-1\\
join_k(\frac{b_k}{2}p-i,V^k_{i\bmod p},W^k_{i\bmod p}) & \frac{b_k}{2}(p-1)\le 
i\le \frac{b_k}{2}p - 1\\
join_k(\lceil\frac{i+1-\frac{b_k}{2}p}{p-1}\rceil,Z^k,W^k_{i\bmod p}) & 
\frac{b_k}{2}p \le i\le p(b_k-1) - (\frac{b_k}{2}-1)\\
\end{cases}
\]
\end{definition}

For example, to create $X^k_1$ we take the first element of $V^k_1$ and the rest
of elements from $U^k_1$ obtaining the sequence $(k-1) \cdot (\pm k)^{p-2} \cdot
(-k) \cdot (\pm k)^{p-1} \cdot (-k) \cdot (\pm k)^{p-1} \cdot (-k) \ldots$ of
length $b_k/2$. In the next lemma we claim that $X^k_1$ really describes the state
after the first stage of computation, where input is a balanced 2-flat sequence.
\begin{lemma}\label{l6}
Let $k\ge p$ and let $\overline{c} = (c_1,\ldots,c_{b_k})$ be a balanced 2-flat
sequence of integers from $[0,2^{k-1}-1]$. Then
$(\hat{Q}^k_1)(reduce(\overline{c})) \in I(X^k_1)$.
\end{lemma}
\begin{proof}
Recall that $H_i=2^i-1$. Let $s = height(\overline{c})$ and $\overline{d} =
(d_1,\ldots,d_{b_k/2}) = reduce(\overline{c})$ By Definitions \ref{flat} and
\ref{reduce} $s=c_i+c_{b_k-i+1}$ and each $d_i = c_i-\frac{s}{2} =
\frac{c_i-c_{b_k-i+1}}{2}$. Observe that each $d_i\in I(\pm k) =
[-\frac{H_{k-1}}{2},\frac{H_{k-1}}{2}]$. We can get this from the following
sequence of inequalities: $-\frac{H_{k-1}}{2} \le \frac{-c_{b_k-i+1}}{2} \le
\frac{c_i-c_{b_k-i+1}}{2} \le \frac{c_i}{2} \le \frac{H_{k-1}}{2}$. Moreover, the
sequence $\overline{d}$ is 2-flat, because $\overline{c}$ is 2-flat. That means
that $d_1\le d_3\le d_5\le \ldots \le d_{k'} \le d_1+1$ and $d_2\le d_4\le d_6\le
\ldots \le d_{k''} \le d_2+1$, where $k' = 2\lceil\frac{b_k}{4}\rceil-1$ and $k''
= 2\lfloor\frac{b_k}{4}\rfloor$.
\begin{fact} \label{fct-13}
Either $-\frac{1}{2}\le d_1$ and $d_{k''}\le 0$ or $-\frac{1}{2}\le d_2$ 
and $d_{k'}\le 0$.
\end{fact}

To prove the fact we consider three cases of the value of $d_1$.

\noindent\textbf{Case $d_1\ge 0$.} In this case we have to prove only that 
$d_{k''}\le 0$. But it is true since $d_{k''}=\frac{c_{k''}-c_{b_k-k''+1}}{2} 
\le \frac{c_{b_k}-c_1}{2} = -d_1 \le 0$. The last inequality holds, because 
$\overline{c}$ is 2-flat and both $k''$ and $b_k$ are even.\\[3pt]
\noindent\textbf{Case $d_1\le -1$.} Then $d_{k'} \le d_1+1 \le 0$. Thus it remains 
to prove that $d_2\ge -\frac{1}{2}$. Similar to the previous case, we 
observe that $d_2 = \frac{c_2-c_{b_k-1}}{2} \ge \frac{c_{b_k}-1-(c_1+1)}{2} = 
-d_1-1 \ge 0$.\\[3pt]
\noindent\textbf{Case $d_1 = -\frac{1}{2}$.} Then $d_{k'} \le d_1+1 =
\frac{1}{2}$ and from $-\frac{1}{2} = \frac{c_1-c_{b_k}}{2}$ we get
$c_1+1 = c_{b_k} \le c_2+1$. Since $c_2 \ge c_1$, we have $d_2 \ge d_1
= -\frac{1}{2}$. If $d_{k'}\le 0$, we are done. Otherwise $d_{k'} =
\frac{1}{2}$ and we have to show that $d_{k''} \le 0$. To this end let
us notice that $\frac{s}{2} = c_1-d_1 = c_1+\frac{1}{2}$ and
$c_{b_k-k'+1} = s-c_{k'} = s-(d_{k'}+\frac{s}{2}) =
\frac{s}{2}-\frac{1}{2} = c_1$. It follows that $c_{k''} = c_1$ since
$c_1 \le c_2 \le c_{k''} \le c_{b_k-k'+1} = c_1$. Thus $d_{k''} = d_1
= -\frac{1}{2}$ and this concludes the proof of Fact \ref{fct-13}.

From Fact \ref{fct-13} and since $\overline{d}$ is 2-flat we can immediately 
get the following corollary.
\begin{corollary}
$\overline{d} \in I((k-1,-k,k-1,-k,\ldots) \cup I(-k,k-1,-k,k-1,\ldots)$.
\end{corollary}

To finish the proof of the lemma we need one more fact:
\begin{fact}
$(\hat{Q}^k_1)(I((k-1, -k, k-1, -k, \ldots) \cup I(-k, k-1,
  -k, k-1, \ldots)) \subseteq I(X^k_1)$.
\end{fact}
Observe, firstly, that the $Cyc$ function is applied to the first position in an
input sequence, thus the input to $Cyc$ is either from $I(k-1)$ or from $I(-k)$.
By Fact \ref{fct-12}.2, $Cyc(I(k-1)) \subseteq I(k-1)$ and $Cyc(-k) \subseteq
I(k-1)$, thus each corresponding output on the first position is correct. On the
other positions in the output sequence we have either $I(\pm k)$ or $I(-k)$ and
$I(-k)$ appears only on positions, which indices are multiples of $p$. If $j$, $1
\le j \le \frac{b_k}{2}$, is a multiple of $p$, then $j \in args(mov^k_j) \in
Q^k_1$ and that means that in $\hat{Q}^k_1$ the $MinMax$ function is applied to
positions $j$ and $j+1$ or the $Min$ function if $j = \frac{b_k}{2}$. In the
former case, on the positions $j$ and $j+1$ in an input sequence, we have a pair
from either $I(-k, k-1)$ or $I(k-1, -k)$. By Fact \ref{fct-12}.4 the output on the
$j$-th position must be from $I(-k)$. In the later case, the $Min$ function is
applied to an element of $I(-k) \cup I(k-1) \subseteq I(\pm k)$. By Fact
\ref{fct-12}.3 we have $Min(I(\pm k)) \subseteq I(-k)$. Finally, on a position $j
> 1$  such that $j \bmod p \notin \{0,1\}$, the input value is from $I(\pm k)$ and
only a $Dec_*$ function can be applied to that value. But $Dec_i(I(\pm k))
\subseteq I(\pm k)$ by Fact \ref{fct-12}.1 and that finishes the proof of the
fact.  \qed
\end{proof}

Informally speaking, the next steps of a computation go as follows: each value
$\pm k$ is moving to its neighbour right position every $p-1$ round with the
help of $MinMax$ function; at the last position the value is changed to $-k$
by $Min$ function; each value $-k$ is moving to the left every round and at
the first position it is changed to $k-1$ by $Cyc$; each value $k-1$ is
decreased by one $p-2$ times at the first position with the help of $Dec_*$
functions, then it is moved to the second position, decreased $p-2$ times
again and so on; at the last position the value is finally decreased from one
to zero by $Dec_0$ or $Min$ and starts moving to the left, one position a
round, stopping at the first position or next to the previous zero. After
$p(b_k-1) - (\frac{b_k}{2}-1)$ rounds the sequence contains only zeroes. See
Figure \ref{states} to observe the initial steps of the process. In the next
lemma we formally describe such computations. To prove it we need one more
technical fact. %
\begin{figure}
\centering
\epsfig{file=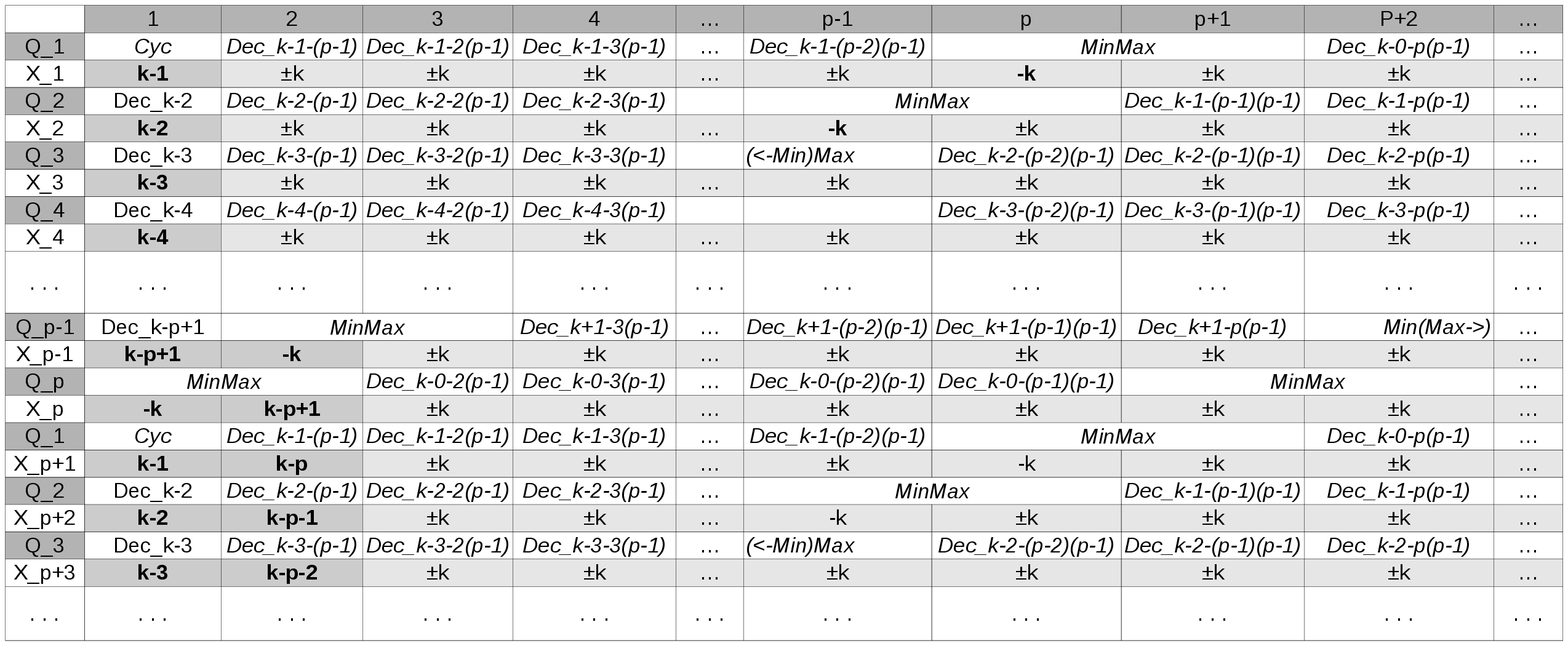, width=4.8in}
\caption{The initial steps of a computation on sequences of interval descriptors.}
\label{states}
\end{figure}
\begin{fact} \label{fct-15}
For all $i$ and $l$ such that $1 < i \le p(b_k-1) - (\frac{b_k}{2}-1)$ and $1 \le
l \le b_k/2$ the pair $(X^k_{i-1,l}, X^k_{i,l})$ is equal to either (a)
$(U^k_{x',l}, U^k_{x,l})$ or (b) $(V^k_{x',l}, V^k_{x,l})$ or (c) $(W^k_{x',l},
W^k_{x,l})$ or (d) $(Z^k_{l}, Z^k_{l})$, where $x' = (i-1) \bmod p$ and $x = i
\bmod p$.
\end{fact}
\begin{proof}
The fact is obviously true for such pair of $i$ and $l$ that both $X^k_{i-1}$ and
$X^k_i$ are defined in the same case of Definition \ref{def-X} and the first
argument of $join_k$ does not change its value between $X^k_{i-1}$ and $X^k_i$.
Thus we have to prove the fact for the following other cases.

\noindent\textbf{Case $1 < i < \frac{b_k}{2}(p-1)$}. We have to consider only $i =
a(p-1)$ and $l = a + 1$ for an integer $a$. Then $x = i \bmod p \equiv -a \pmod p$
and $x + l \equiv 1 \pmod p$ By Definition \ref{def-X}, we have $X^k_{i-1,l} =
U^k_{x',l}$ and $X^k_{i,l} = V^k_{x,l}$. By Definition \ref{def-seq}, $V^k_{x,l} =
U^k_{x,l}$, thus we are in case (a) of the fact.

\noindent\textbf{Case $i = \frac{b_k}{2}(p-1)$}. Then $X^k_{i-1} =
join_k(b_k/2, V^k_{x'}, U^k_{x'}) = V^k_{x'}$ and $X^k_i = join_k(b_k/2,
V^k_x,W^k_x) = V^k_x$ by the definition. It follows that for all values of $l$
we get case (b) of the fact.

\noindent\textbf{Case $\frac{b_k}{2}(p-1) < i < \frac{b_k}{2}p$}. We have to
consider only $l = \frac{b_k}{2}p - i$. If $x + l \not\equiv 1 \pmod p$ then
$W^k_{x,l} = V^k_{x,l}$, by Definition \ref{def-seq}, and we are again in case (b)
of the fact. Otherwise $x + l \equiv 1 \pmod p$, but $x' + l \not\equiv 1 \pmod
p$. In this case $X^k_{i-1} = V^k_{x',l} = W^k_{x',l}$ and $X^k_i = W^k_{x,l}$. 
Thus we get case (c) of the fact.

\noindent\textbf{Case $i = \frac{b_k}{2}p$}. Then $X^k_{i-1} = join_k(1, V^k_{x'},
W^k_{x'})$ and $X^k_i = join_k(1, Z^k,W^k_x)$ by the definition. The only case we
have to check is $l = 1$. In this case $x + l \equiv 1 \pmod p$ and, by Definition 
\ref{def-seq}, $W^k_{x,1} = 0 = Z^k_1 = X^k_{i,1}$. In addition, we have $x' + l
\not\equiv 1 \pmod p$, therefore $X^k_{i-1,1} = V^k_{x',1} = W^k_{x',1}$ and we
have case (c) of the fact.

\noindent\textbf{Case $\frac{b_k}{2}p < i \le p(b_k-1) - (\frac{b_k}{2}-1)$}. In 
this last case we have to consider only $i = \frac{b_p}{2}p + b(p-1)$ and $l = b+1$ 
for an integer $b$. Then $x \equiv -b \pmod p$ and, consequently, $x +l \equiv 1 
\pmod p$. It follows that $X^k_{i,l} = Z^k_l = 0 = W^k_{x,l}$ and, by Definition 
\ref{def-X}, $X^k_{i-1,l} = W^k_{x',l}$, so we are again in case (c) of the fact. 
\qed
\end{proof}
In the next key lemma of this subsections we claim that state sequences defined in
Definition \ref{def-X} really describe any computation on intervals assuming that
we start with a balanced 2-flat sequence. 
\begin{lemma}\label{l7}
For $k\ge p$ and each $i=1,2,\ldots,p(b_k-1)-\frac{b_k}{2}$ the following 
inclusion holds:
\[ (\hat{Q}^k_{i\bmod p +1})(I(X^k_i))\subseteq I(X^k_{i+1}). \]
\end{lemma}
\begin{proof}
We have to prove, equivalently, that for $k\ge p$ and $x = 1, \ldots, p$ the
following inclusions are true: $( \hat{Q}^k_x)(I(X^k_{pj+x-1})) \subseteq
I(X^k_{pj+x})$, where $j=1, \ldots\ \lfloor\frac{p(b_k-1)-b_k/2}{p}\rfloor$ for
$x=1$ and $j=0,1, \ldots \lfloor\frac{p(b_k-1)-b_k/2-x+1}{p}\rfloor$ for $x > 1$.
The value of the function $\hat{Q}^k_x$, $x=1, \ldots, p$, on a fixed position can
be computed with the help of one of the functions $Cyc$, $Dec_*$, $MinMax$ and
$Min$ introduced in Definition \ref{def-Cyc} (see also Fact \ref{fct-reduce}). We
consider these functions one after another analysing which positions in state
sequences are modified by them and what values are in that positions before and
after applying a function. In the following, we denote by $A_{i,j}$ the $j$-th
element of a sequence $A_i$.

The function $Cyc$ corresponds to $\hat{cyc}^k$, which is used only in the
definition of $\hat{Q}^k_1$ and modifies just the position 1 of the sequences
$I(X^k_{pj})$, where $j=1, \ldots\ \lfloor\frac{p(b_k-1)-b_k/2}{p}\rfloor$. Thus
it is enough to show the inclusion $Cyc(I(X^k_{pj,1}))\subseteq I(X^k_{pj+1,1})$.
By Definition \ref{def-X} the argument of $Cyc\cdot I$ can be either $X^k_{pj,1} =
V^k_{0,1} = -k$ for $pj < \frac{b_k}{2}p$ or $X^k_{pj,1} = Z_1 = 0$ for $pj \ge
\frac{b_k}{2}p$. The corresponding value of the next state sequence is
$X^k_{pj+1,1} = V^k_{1,1} = k-1$ for $pj+1 < \frac{b_k}{2}p$ or $X^k_{pj+1,1} =
Z_1 = 0$ for $pj+1 \ge \frac{b_k}{2}p $. Using Fact \ref{fct-12}, both inclusions
$Cyc(I(-k)) \subseteq I(k-1)$ and $Cyc(I(0)) \subseteq I(0)$ are true and we are
done.

In the set $\hat{Q}^k_x$ there are several $\hat{dec}^k_{l,s}$ functions, each of
which satisfies the conditions $x + l \not\equiv 1, 2 \pmod p$, $1 \le l \le
b_k/2$ and $s = (p-2)(l-1) - 1 - (x + l - 1) \bmod p \le k-1$. We know also that
$args(\hat{dec}^k_{l,s}) = \{l\}$ and $(\hat{dec}^k_{l,s}(\overline{d}))_l =
Dec_{k-s-1}(d_l)$ for a sequence $\overline{d} = (d_1, \ldots, d_{b_k/2})$, thus we
can rewrite our proof goal for that functions as the following fact.
\begin{fact}
Let $k\ge p$. For any $x$, $l$ and $s$ such that $1 \le x \le p$, $1 \le l \le
b_k/2$, $x + l \not\equiv 1,2 \pmod p$ and $1 \le s = (p-2)(l-1) - 1 - (x + l - 1)
\bmod p \le k -1$ we have
\[Dec_{k-s-1}(I(X^k_{pj+x-1,l})) \subseteq I(X^k_{pj+x,l}),\] 
for any $j\ge 0$ such that state sequences  $X^k_{pj+x-1}$
and $X^k_{pj+x}$ are defined.
\end{fact}

The sequences $X^k_*$ are defined with the help of sequences $U^k_*$, $V^k_*$,
$W^k_*$ and $Z_*$. In $U^k_x$, $V^k_x$ and $W^k_x$ there are {\em strange} 
``moving-left'' elements $-k$ or $0$ that appears on positions whose indices 
$\equiv -x+1 \pmod p$. Thus those strange elements cannot appear on position $l$ in 
$X^k_{pj+x-1}$ and $X^k_{pj+x}$, since, otherwise, $l \equiv -x+1 \pmod p$ or $l 
\equiv -(x-1)+1 \pmod p$, but we know that $l+x \not\equiv 1,2 \pmod p$. By Fact 
\ref{fct-15}, we have to consider just the following three cases of 
values $X^k_{pj+x-1,l}$ and $X^k_{pj+x,l}$. 
\setlength{\tabcolsep}{3pt}
\begin{center}
\begin{tabular}{||l|l||c|c||}
\hline\hline \multicolumn{1}{||c|}{Cases of} & 
      \multicolumn{1}{|c||}{Cases of} & Value of  & Value of   \\ 
      $y=X^k_{pj+x-1,l}$  & $y'=X^k_{pj+x,l}$  & $y$  & $y'$  \\ \hline
\hline $y=U^k_{x-1,l}$ & $y'=U^k_{x,l}$ & $\pm k$  & $\pm k$  \\ 
\cline{1-4} $y=V^k_{x-1,l}=W^k_{x-1,l}$ & $y'=V^k_{x,l}=W^k_{x,l}$ & $k-s$ & 
$k-s-1$ \\ 
\cline{1-4} $y=Z_{l}$ & $y'=Z_{l}$ & 0 & 0  \\ 
\hline\hline 
\end{tabular}
\end{center}
In all cases above we have $Dec_{k-s-1}(I(y)) \subseteq I(y')$ by Fact
\ref{fct-12}.1. Since it it not obvious that $V^k_{x-1,l} = k-s$ and $V^k_{x,l} =
k - s - 1$ (the second case in the table), we prove these equations now. Since $x
+ l \not\equiv 1,2 \pmod p$ and $s \le k-1$, it follows that $V^k_{x-1,l} =
e^k_p(x-1,l) = \max(0, k - (p-2)(l-1) - (x-1 + l - 1) \bmod p) = k -
\min(k,(p-2)(l-1) - 1 + (x + l - 1) \bmod p) = k - s$. In a similar way,
$V^k_{x,l} = e^k_p(x,l) = \max(0, k - (p-2)(l-1) - (x + l - 1) \bmod p) = k - 1 -
\min(k-1, -1 + (p-2)(l-1) + (x + l - 1) \bmod p) = k - s - 1$.

The next function to be analysed is $MinMax$. It corresponds to all
$\hat{mov}^k_{l}$ functions in a set $\hat{Q}^k_x$, where $1 \le l < b_k/2$ and $1
\le x \le p$. By Definitions \ref{defQ} and \ref{def-Cyc}, each such function
satisfies the condition $l+x \equiv 1 \pmod p$. We know also that
$args(\hat{mov}^k_{l}) = \{l,l+1\}$ and, by Fact \ref{fct-reduce}
$(\hat{mov}^k_{l}(\overline{d}))_{l+j} = (MinMax(d_l,d_{l+1}))_{1+j}$ for $j \in
\{0,1\}$ and any sequence $\overline{d} = (d_1, \ldots, d_{b_k/2})$. Thus, to
prove the lemma, it suffices to show the following fact.
\begin{fact}
Let $k\ge p$. For any $x$ and $l$ such that $1 \le x \le p$, $1 \le l < b_k/2$ and
 $l+x \equiv 1 \pmod p$ we have
\[MinMax(I(X^k_{pj+x-1,l},X^k_{pj+x-1,l+1})) \subseteq I(X^k_{pj+x,l},X^k_{pj+x,l+1}),\]
where $j \ge 0$ is an integer such that state sequences  $X^k_{pj+x-1}$
and $X^k_{pj+x}$ are defined.
\end{fact}

As we do with the previous functions, we prove the fact by considering all possible
cases in the following table. All of its values are set according to Definition
\ref{def-seq}, since $x + l = x -1 + l + 1 \equiv 1 \pmod p$. To reduce the size
of the table we also use the following shortcuts: $a=pj+x$ and $y = k - (p-2)l - 1$.
\begin{center}
\begin{tabular}{|*{4}{|l}|*{4}{|c}||}
\hline\hline \multicolumn{2}{||c|}{Cases of $(s_1,s_2)$} & 
      \multicolumn{2}{|c||}{Cases of $(t_1,t_2)$} & 
      \multicolumn{2}{|c|}{Value of} & 
      \multicolumn{2}{|c||}{Value of} \\ 
  $s_1 = X^k_{a-1,l}$  & $s_2 = X^k_{a-1,l+1}$  & $t_1 = X^k_{a,l}$  & $t_2 = 
  X^k_{a,l+1}$  & $s_1$  & $s_2$  & $t_1$  & $t_2$ \\ \hline
\hline $U^k_{x-1,l}$ & $U^k_{x-1,l+1}$ & $U^k_{x,l} = V^k_{x,l}$ & $U^k_{x,l+1}$ &  
          $\pm k$    &       $-k$      &    $-k$     &    $\pm k$ \\
\hline $V^k_{x-1,l}$ & $V^k_{x-1,l+1} = U^k_{x-1,l+1}$ & $V^k_{x,l}$ & $V^k_{x,l+1} 
= W^k_{x,l+1}$ & $y$ &       $-k$      &   $-k$      &    $y$ \\ 
\hline $W^k_{x-1,l} = V^k_{x-1,l}$ & $W^k_{x-1,l+1} = Z_{l+1}$ & $W^k_{x,l}$ & 
$W^k_{x,l+1}$ &            $y$     &       $0$       &       $0$   &     $y$    \\ 
\hline $Z_{l}$ & $ Z_{l+1} = W^k_{x-1,l+1}$ & $Z_{l}$ & $Z_{l+1}$ & 
         $0$   &       $0$                   &   $0$   &    $0$   \\ 
\hline\hline 
\end{tabular}
\end{center}
In all cases above we have $MinMax(I(s_1,s_2)) \subseteq I(t_1,t_2)$ by Facts
\ref{fct-12}.4 and \ref{fct-12}.5. Thus, to end the proof of the fact we have to
check whether $V^k_{x-1,l} = V^k_{x,l+1} = y$. From the definition $V^k_{x-1,l} =
e^k_p(x-1,l) = \max(0, k - (p - 2)(l - 1) - (x-1+l-1) \bmod p) = \max(0, k -
(p-2)l + (p-2) - (p-1)) = \max(0, k - (p-2)l - 1) = y$. The equality $(x-1+l-1)
\bmod p = p-1$ follows from $x + l \equiv 1 \pmod p$. We use also the fact that
from $l < b_k/2$ we can get $(p-2)l + 1 \le k$. In the same way, $V^k_{x,l+1} =
e^k_p(x,l+1) = \max(0, k - (p - 2)l - (x+l+1-1) \bmod p) = y$.

The last function to be considered is $Min$. It corresponds to all $\hat{mov^k_l}$
functions in $\hat{Q}^k_x$, $1 \le x \le p$, such that $x + l \equiv 1 \pmod p$
and $l = (k-2)/(p-2)$. Thus, to finish the proof of the lemma, it suffices to show
the following fact.
\begin{fact}
Let $k\ge p$. For all $1 \le x \le p$ and $l = (k-2)/(p-2)$ such that 
$x + l \equiv 1 \pmod p$ we have 
\[ Min(I(X^k_{pj+x-1,l})) \subseteq I(X^k_{pj+x,l}),\] 
where $j \ge 0$ is an integer such that state sequences  $X^k_{pj+x-1}$
and $X^k_{pj+x}$ are defined.
\end{fact}
As in the case of previous functions we prove the fact by considering all 
possible cases in the following table. 
\begin{center}
\begin{tabular}{||l|l||c|c||}
\hline\hline
      \multicolumn{1}{||c|}{Cases of} & \multicolumn{1}{|c||}{Cases of} & 
      Value of  & Value of  \\ 
      $s=X^k_{pj+x-1,l}$  & $t=X^k_{pj+x,l}$  & $s$  & $t$  \\ \hline
\hline $s = U^k_{x-1,l}$ & $t = U^k_{x,l} = V^k_{x,l}$ & $\pm k$ & 
$-k$  \\ 
\hline $s = V^k_{x-1,l} = W^k_{x-1,l}$ & $t = W^k_{x,l} = Z_{l}$ & 
$1$ & 
$0$ \\ 
\hline $s = Z_{l}$ & $t = Z_{l}$ & 0 & $0$ \\ 
\hline\hline 
\end{tabular}
\end{center}
In all cases above we have $Min(I(s)) \subseteq I(t)$ by Fact \ref{fct-12}.3.
Observe that $l + x - 1 \equiv 0 \pmod p$, thus we have to check whether 
$V^k_{x-1,l} = 1$. By the definition $V^k_{x-1,l} = e^k_p(x-1,l)  = \max(0, k - 
(p-2)(l-1) - (x -1 + l -1) \bmod p) = \max(0, k - (p-2)l + p - 2 - (p-1)) = \max(0, 
k - (k-2) - 1) = 1$. \qed
\end{proof}
\begin{lemma} \label{merge-bal}
Let $k\ge p$, $D = p(b_k - 1) - (\frac{b_k}{2} - 1)$ and let $\overline{c} =
(c_1,\ldots,c_{b_k})$ be a balanced 2-flat sequence of integers from
$[0,2^{k-1}-1]$ and let $s=height(\overline{c})$. Let $f = f_{D}\circ f_{D-1}\circ
\ldots\circ f_1$, where $f_i = Q^k_{((i-1)\bmod 3)+1}$, $i=1,\ldots,D$. Then
$f(\overline{c}) = (\frac{s}{2})^{b_k}$ if $s$ is even or $f(\overline{c}) =
(\frac{s-1}{2})^{b_k/2}\cdot(\frac{s+1}{2})^{b_k/2}$ otherwise.
\end{lemma}
\begin{proof}
Since each $f_i$ maps a balanced sequence to a balanced one, let $\hat{f}_i
= reduce(f_i,s) = \hat{Q}^k_{((i-1)\bmod p)+1}$, where the later
equality follows from Lemma \ref{l5}. Let also $\overline{d}_0 =
reduce(\overline{c})$ and let $\overline{d}_i = \hat{f}_i(\overline{d}_{i-1})$
for $i=1,\ldots,D$. Then $\overline{d}_1\in I(X^k_1)$ by Lemma \ref{l6}
and for $i=2,\ldots,D$ we get $\overline{d}_i\in I(X^k_i)$ by an easy
induction and Lemma \ref{l7}. Let $\mathbb{Z}$ denote, as usual, the set of 
integers. By $\mathbb{Z}_{\frac{1}{2}}$ we will denote the set 
$\{z+\frac{1}{2}| z\in\mathbb{Z}\}$. Looking at Definitions \ref{reduce} and 
\ref{def-Cyc} observe the 
following fact:
\begin{fact}
If $s$ is even then all elements of sequences $\overline{d}_i$, $i=0,\ldots, 
D$, are integers. If $s$ is odd then  all elements of sequences 
$\overline{d}_i$, $i=0,\ldots, D$, are in $\mathbb{Z}_{\frac{1}{2}}$.
\end{fact}
Since $\overline{d}_{D} \in I(X^k_{D}) = I(0^{b_k/2})$ and $I(0)\cap 
\mathbb{Z} = \{0\}$ and $I(0)\cap \mathbb{Z}_{\frac{1}{2}} = 
\{\frac{1}{2}\}$, it follows that $\overline{d}_{D} = 0^{b_k/2}$ if $s$ is 
even and  $\overline{d}_{D} = \frac{1}{2}^{b_k/2}$, otherwise. Using now 
the definition of $s$-extended sequence to $0^{b_k/2}$ and $\frac{1}{2}^{b_k/2}$ 
we get the desired conclusion of the lemma. \qed
\end{proof}
In this way, with respect to Lemma \ref{l3}, we have proved that the network $M_k$
is able to merge in $p(b_k - 1) - (\frac{b_k}{2} - 1)$ stages two sorted sequences
given in odd and even registers, provided that the numbers of ones in our matrix
columns form a balanced sequence. If the sequence is not balanced, $\frac{b_k}{2}
- 1$ additional stages are needed to get a sorted output.

\subsection{Analysis of General Columns}\label{ss:3} 

In a general case we will use balanced sequences as lower and upper bounds on
the numbers of ones in our matrix columns and observe that $Q^k_x$, $1 \le x \le 
p$, are monotone functions (see Fact \ref{fct-10}).
\begin{definition}\label{def-14}
Let $k\ge p$ and let $\overline{c} = (c_1,\ldots,c_{b_k})$ be a
2-flat sequence of integers from $[0,2^{k-1}-1]$ that is not balanced.
Since both $\overline{c}_{odd} =(c_1,\ldots,c_{b_k-1})$ and 
$\overline{c}_{evn} =(c_2,\ldots,c_{b_k})$ are 
flat sequences, let $i$ ($j$, respectively) be such that $c_{2i-1} < 
c_{2i+1}$ 
($c_{b_k-2j} < c_{b_k-2j+2}$, respectively) or let $i=k-2$ ($j=k-2$) if 
$\overline{c}_{odd}$ ($\overline{c}_{evn}$, respectively) is a constant 
sequence. The defined below sequences $\check{c}$ and $\hat{c}$ we will call 
lower and upper bounds of $\overline{c}$. If $i < j$ then for $l=1, \ldots, 
b_k$ 
\begin{align*}
\check{c}_l = \begin{cases}
c_1       &\text{if $l$ is odd and $l\le 2j-1$}\\
c_{b_k-1}   &\text{if $l$ is odd and $l\ge 2j+1$}\\
c_l       &\text{if $l$ is even}\\
\end{cases}
~&~
\hat{c}_l = \begin{cases}
c_{b_k-1} &\text{if $l$ is odd}\\
c_{b_k}   &\text{if $l$ is even}\\
\end{cases}
\end{align*} 
If $i > j$ then for $l=1, \ldots, b_k$ 
\begin{align*}
\check{c}_l = 
\begin{cases}
c_1      &\text{if $l$ is odd}\\
c_2      &\text{if $l$ is even}\\
\end{cases}
~&~
\hat{c}_l = 
\begin{cases}
c_l       &\text{if $l$ is odd}\\
c_2       &\text{if $l$ is even and $l\le b_k-2i$}\\
c_{b_k}   &\text{if $l$ is even and $l >  b_k-2i$}\\
\end{cases}
\end{align*} 
\end{definition}
\begin{fact}\label{fct-19}
For $k\ge p$ and any not balanced 2-flat sequence $\overline{c} =
(c_1,\ldots,c_{b_k})$ of integers from $[0,2^{k-1}-1]$ the sequences
$\check{c}$ and $\hat{c}$ are balanced, $height(\check{c})+1 =
height(\hat{c})$ and $\check{c}\le\overline{c}\le\hat{c}$.
\end{fact}
\begin{proof}
Let $i$ and $j$ be defined as in Definition \ref{def-14}. We will consider 
only the case $i<j$. The proof of the other case is similar. Directly 
from the definition we get that $\hat{c}$ is balanced. To see that 
$\check{c}$ is also balanced let us check for $l=1,\ldots,b_k/2$ whether the 
sum $\check{c}_{2l-1} + \check{c}_{b_k-2l+2}$ is constant.
\[
\check{c}_{2l-1} + \check{c}_{b_k-2l+2} = \check{c}_{2l-1} + c_{b_k-2l+2}
= \begin{cases}
c_1 + c_{b_k-2l+2} = c_1 + c_{b_k}  & \text{if  $l\le j$}\\
c_{b_k-1} + c_{b_k-2l+2} = c_{b_k-1} + c_2 & \text{otherwise}
\end{cases}
\]
If $j = k-2$ there is no otherwise case and we are done. If $j < k-2$ then 
$c_{b_k}-c_2 = c_{b_k-1}-c_1 = 1$, because of the definition of $i$ and $j$ 
and we are also done. Moreover $height(\check{c}) + 1 = c_1 + c_{b_k} +1 = 
c_{b_k-1} + c_{b_k} = height(\hat{c})$. To prove that $\check{c}\le 
\overline{c}\le \hat{c}$ we consider even and odd indices. For even indices 
from the definition we have: $\check{c}_{2l} = c_{2l} \le c_{b_k} = 
\hat{c}_{2l}$. For odd indices $\hat{c}_{2l-1} = c_{b_k-1} \ge c_{2l-1} \ge 
c_1$. If $l\le j$ we are done, otherwise, $c_{2l-1} = c_{b_k-1} = 
\check{c}_{2l-1}$, because $\overline{c}_{odd}$ is flat. \qed
\end{proof}
\begin{theorem}\label{thm-19}
Let $k\ge p$, $D = p(b_k - 1)$ and let $\overline{c} = (c_1,\ldots,c_{b_k})$ be a
2-flat sequence of integers from $[0,2^{k-1}-1]$. Let $f = f_{D}\circ f_{D-1}\circ
\ldots\circ f_1$, where $f_i = Q^k_{((i-1)\bmod 3)+1}$, $i=1,\ldots,D$. Then
$f(\overline{c})$ is a flat sequence.
\end{theorem}
\begin{proof}
Let $\overline{c}$ be a 2-flat sequence of integers from $[0,2^{k-1}-1]$. If
$\overline{c}$ is balanced then $f(\overline{c})$ is a flat sequence due to Lemma
\ref{merge-bal} and an observation that flat sequences are not modified by $Q^k_*$ 
functions. Otherwise, let $\check{c}$ and $\hat{c}$ be its balanced lower
and upper bounds, as defined in Definition \ref{def-14}. Let $\overline{c}_0 =
\overline{c}$, $\check{c}_0 = \check{c}$, $\hat{c}_0 = \hat{c}$ and for $i = 1,
\ldots, D$ let us define $\overline{c}_i = f_i(\overline{c}_{i-1})$, $\check{c}_i
= f_i(\check{c}_{i-1})$ and $\hat{c}_i = f_i(\hat{c}_{i-1})$. Observe that
$\check{c}_i\le \overline{c}_i\le \hat{c}_i$, because of monotonicity of functions
$Q^k_x$, $1 \le x \le p$, and Fact \ref{fct-19}. To prove that $\overline{c}_{D}$
is a flat sequence we need the following three technical facts.
\begin{fact}\label{fct-21}
Let $s = height(\check{c})$. If $s$ is even then $\overline{c}_{i,j} =
\frac{s}{2}$ and $\overline{c}_{i,b_k-j+1}\in\{\frac{s}{2},\frac{s}{2}+1\}$
for each $i$ and $j$ such that $p b_k/2 \le i \le D - (b_k/2 -1)$ and $1 \le j \le
\lceil\frac{i+1-(p b_k/2)}{2}\rceil$. If $s$ is odd then  $\overline{c}_{i,j}
\in\{\frac{s-1}{2},\frac{s+1}{2}\}$ and $\overline{c}_{i,b_k-j+1} =
\frac{s+1}{2}$ for each $i$ and $j$ such that $p b_k/2 \le i \le D - (b_k/2 - 1)$ 
and $1 \le j \le \lceil\frac{i+1-(p b_k/2)}{2}\rceil$. 
\end{fact}
\begin{proof}
Since both $\check{c}$ and $\hat{c}$ are balanced, we can consider reduced forms
of them and use Lemmas \ref{l6} and \ref{l7}. For the given range of $i$'s values
that means that both $reduce(\check{c}_i)$ and $reduce(\hat{c}_i)$ are in
$I(X^k_i)$. Recall that $I(X^k_i) = I(join_k(\lceil\frac{i+1-(p
b_k/2)}{2}\rceil,Z^k,W^k_i))$ for $i = p b_k/2, \ldots, D - (b_k/2 - 1)$. Hence
for a given range of $j$'s values both $reduce(\check{c}_i)_j$ and
$reduce(\hat{c}_i)_j$ are in $I(0) = [-\frac{1}{2},0]$. From Fact \ref{fct-19} we
know that $height(\hat{c}) = s+1$ and from Lemma \ref{l4} that heights are
preserved in sequences $\check{c}_i$ and $\hat{c}_i$. Thus, from the definition of
a reduced sequence, $\check{c}_{i,j} \in [\frac{s-1}{2},\frac{s}{2}]$,
$\check{c}_{i,b_k-j+1} \in [\frac{s}{2},\frac{s+1}{2}]$, $\hat{c}_{i,j} \in
[\frac{s}{2},\frac{s+1}{2}]$ and $\hat{c}_{i,b_k-j+1} \in
[\frac{s+1}{2},\frac{s+2}{2}]$. Since all $\check{c}_i$ and $\hat{c}_i$ are
sequences of integers, for even $s$ we get $\check{c}_{i,j} =
\check{c}_{i,b_k-j+1} = \hat{c}_{i,j} = \frac{s}{2}$ and $\hat{c}_{i,b_k-j+1} =
\frac{s+2}{2}$; for odd $s$ we conclude that $\check{c}_{i,j} = \frac{s-1}{2}$ and
$\check{c}_{i,b_k-j+1} = \hat{c}_{i,j} = \hat{c}_{i,b_k-j+1} = \frac{s+1}{2}$.
Since $\check{c}_{i,j} \le \overline{c}_{i,j} \le \hat{c}_{i,j}$, the fact
follows. \qed
\end{proof}
The second fact extends the first fact up to the last stage of our 
computation.
\begin{fact}\label{fct-22}
Let $s = height(\check{c})$. If $s$ is even then $\overline{c}_{i,j} =
\frac{s}{2}$ and $\overline{c}_{i,b_k-j+1}\in\{\frac{s}{2},\frac{s}{2}+1\}$ for
each $i = D-(b_k/2-2), \ldots, D$ and $j = 1, \ldots, b_k/2$. If $s$ is odd then
$\overline{c}_{i,j} \in\{\frac{s-1}{2},\frac{s+1}{2}\}$ and
$\overline{c}_{i,b_k-j+1} = \frac{s+1}{2}$ for each $i = D-(b_k/2-2), \ldots, D$
and $j = 1, \ldots, b_k/2$.
\end{fact}
\begin{proof}
Consider first the sequence $\overline{c}_{D-(b_k/2-1)}$ and observe that for $i =
D-(b_k/2-1)$ the value of $\lceil\frac{i+1-(p b_k/2)}{p-1}\rceil$ is equal to
$b_k/2$. It follows from Fact \ref{fct-21} that for even $s$ all values from the
left half of $\overline{c}_{D-(b_k/2-1)}$ are equal to $\frac{s}{2}$ and all
values from the right half of $\overline{c}_{D-(b_k/2-1)}$ are in $\{\frac{s}{2},
\frac{s}{2}+1\}$. For odd $s$ all values from the left half of
$\overline{c}_{D-(b_k/2-1)}$ are in $\{\frac{s-1}{2},\frac{s+1}{2}\}$ and all
values from the right half of $\overline{c}_{D-(b_k/2-1)}$ are equal to
$\frac{s+1}{2}$. Since $Q^k_x$, $1 \le x \le p$, are built of functions
$dec^k_*$, $mov^k_*$ and $cyc^k$ (cf. Definitions \ref{defFun} and \ref{defQ})
observe that each function $f_i$, $i = D-(b_k/2-2), \ldots, D$ can exchange only
the values at positions from $args(mov^k_*)$ that are from non-constant half of
arguments (in case of $dec^k_*$ and $cyc^k$ we can observe that for $a\le b\le
a+1$ and any $h\ge 0$ we have $\min(a,b+h) = a$, $\max(a-h,b) = b$, $\max(a,b-1) =
a$ and $\min(a+1,b) = b$, that is, the functions are identity mappings in stages
$D-(b_k/2-2), \ldots, D$). The $mov^k_*$ functions can exchange only unequal
values at neighbor positions moving the smaller value to the left. \qed
\end{proof}
The last fact states that unequal values $\overline{c}_{i,j}$ described in the
previous two facts are getting sorted during the last stages of the
computation. Observe that if $s$ is odd (even, respectively) then we have to
trace the sorting process only in a left (right, respectively) region of
indices $[1,\min(b_k/2,\lceil\frac{i+1-(p b_k/2)}{p-1}\rceil)]$ ($[\max(b_k/2
+ 1, b_k-\lceil\frac{i+1-(p b_k/2)}{p-1}\rceil+1),b_k]$, respectively), where
$i = p b_k/2, \ldots, D$ and the values to be sorted differs at most by one.
The other part is already sorted. We trace the positions of the smaller values
$s'=\frac{s-1}{2}$ in the left region and the greater values
$s'=\frac{s}{2}+1$ in the right region. We will call each such  $s'$ a moving
element. For $t = 1, \ldots, b_k/2$ let us define $i_t = p b_k/2 + (p-1)(t-1)$
to be the stage, after which the length of the region extends from $t-1$ to
$t$ and a new element appears in it. Let $t'=t$ for odd $s$ and $t'=b_k-t+1$,
otherwise, be the position of this new element and $a_t=c_{i_t,t'}$ be its
value. Finally, let $n_t = |\{1\le l\le t | a_l = s'\}|$ be the number of
moving elements in the region after stage $i_t$.
\begin{fact}\label{fct-23}
Using the above definitions, for $t = 1, \ldots, b_k/2$, if $a_t = s'$ then for $i
= 0, \ldots, D - i_t$ we have $c_{i_t+i,\max(t-i,n_t)} = a_t$ if $s$ is 
odd and $c_{i_t+i,\min(t'+i,b_k-n_k+1)} = a_t$, otherwise.
\end{fact}
\begin{proof}
We prove the fact only for odd $s$, that is, for the left region. The proof
for the right region is symmetric. We would like to show that if $a_t = s'$
appears at position $t' = t$ after stage $i_t$ then it moves in each of the
following stages one position to the left up to its final position $n_t$. The
proof is by induction on $t$ and $i$. If $t=1$ and $a_1=s'$ appears at
position 1 after stage $i_1=p b_k/2$ then $n_1=1$ and $a_1$ is already at its
final position. It never moves, because values at second position are $\ge
s'$, by Facts \ref{fct-21} and \ref{fct-22}. If $t>1$ and $a_t = s'$ then the
basis $i=0$ is obviously true. In the inductive step $i>0$ we assume that 
$c_{i_t+i-1,\max(t'-i+1,n_t)} = a_t$ and that the fact is true for smaller 
values of $t$. If $\max(t-i+1,n_t) = n_t$ then also $\max(t-i,n_t) = n_t$ 
and, by the induction hypothesis, values at positions $1, \ldots, n_t-1$ 
are all equal $s'$. That means that $a_t$ is at its final position and we are 
done. Thus we left with the case: $n_t < t-i+1$, that is, with $n_t \le t-i$. 

Consider the sequences $\overline{c}_{i_t+i-1}$ and $\overline{c}_{i_t+i} =
f_{i_t+i}(\overline{c}_{i_t+i-1})$. We know that $\overline{c}_{i_t+i-1,t-i+1} =
s'$. To prove that $\overline{c}_{i_t+i,t-i} = s'$ we would like to show that $s'$
is moved one position to the left by $f_{i_t+i}$, i.e. that
$\overline{c}_{i_t+i-1,t-i} = s'+1$ and $mov^k_{t-i}\in f_{i_t+i}$. The later is a
direct consequence of an observation that $mov^k_a\in f_b$ if and only if
$(a+b)\equiv 1 \pmod p$. In our case $(t-i) + (i_t+i) = t+i_t = t + p b_k/2 +
(p-1)(t-1) = p(b_k/2 + t - 1) +1 \equiv 1 \pmod p$. To prove the former, let us
consider any $a_u = s'$, $u\le t-1$. Then $i_u \le i_t-(p-1) \le i_t - 2$ and $n_u
\le n_t-1$. By the induction hypothesis, $c_{i_u+j,\max(u-j,n_u)} = s'$. Setting
$j = i_t - i_u + i-1$ we get $j\ge i+1$ and $\max(u-j,n_u) \le
\max(t-1-(i+1),n_t-1) < \max(t-i,n_t) = t-i$. Moreover, $i_u+j = i_t+i-1$. That
means that in the sequence $\overline{c}_{i_t+i-1}$ none of $n_t$ elements $s'$ is
at position $t-i$ and, consequently, $\overline{c}_{i_t+i-1,t-i} = s'+1$. Since
$mov^k_{t-i}$ switches $s'$ with $s'+1$, this completes the proof of Fact
\ref{fct-23}. \qed
\end{proof}

Now we are ready to prove that $\overline{c}_{D}$ is a flat sequence. By Fact
\ref{fct-22}, if $s$ is odd then $\overline{c}_{D} \in
\{\frac{s-1}{2},\frac{s+1}{2}\}^{b_k/2}(\frac{s+1}{2})^{b_k/2}$, otherwise,
$\overline{c}_{D} \in
(\frac{s}{2})^{b_k/2}\{\frac{s}{2},\frac{s}{2}+1\}^{b_k/2}$. The number of
minority (moving) elements in $\overline{c}_{D}$ has been denote by
$n_{b_k/2}$. If $s$ is odd and $a_t$, $t = 1, \ldots, b_k/2$, is a minority
element $\frac{s-1}{2}$, then, by Fact \ref{fct-23}, $c_{D,n_t} =
\frac{s-1}{2}$. If $s$ is even and $a_t$, $t = 1, \ldots, b_k/2$, is a
minority element $\frac{s}{2}+1$, then, by Fact \ref{fct-23}, $c_{D,b_k-n_t+1}
= \frac{s}{2}+1$. In both cases this proves that  $\overline{c}_{D}$ is flat,
which completes the proof of Theorem \ref{thm-19}. \qed
\end{proof}

\subsection{Proof of Theorem \ref{3merger}} 

Theorem \ref{3merger} follows directly from Theorem \ref{thm-19} and Lemma
\ref{l3}. Let $k\ge p \ge 4$, $b_k = 2\lceil\frac{k-2}{p-2}\rceil$ and
$\overline{c}$ be any 2-flat sequence of integers from $[0,2^{k-1}-1]$. By
Theorem \ref{thm-19} the result of application $(Q^k_p \circ Q^k_{p-1} \circ
\ldots \circ Q^k_1)^{b_k - 1}$ to $(\overline{c})$ is a flat sequence. Then,
by Lemma \ref {l3}, the network $M_k$ is a $b_k - 1$-pass merger of two sorted
sequences given in odd and even registers, respectively.

\section{Average sorting times} \label{sec4} 
It is easy to observe that our $M^p_k$ networks are also periodic sorters, because
they contain all neighbour conparators $[i:i+1]$, $1\le i< N^p_k$. We were curious
how efficient periodic sorters were they, when the early stopping property would
be applied, that is, when a periodic application of $M^p_k$ would be stopped just
after none of comparator in  $M^p_k$ exchange values. We measured the average and
maximal sorting times (the number of rounds) of $10^5$ (pseudo)random permutations
on selected $M^p_k$ networks, $3 \le p \le 5$, $9 \le k \le 14$ and the results
are shown in Fig. \ref{sort_times}. Surprisingly, the average and maximal sorting
times are quite close to $\log^2(N^p_k)$. An open question is what is the
worst-case sorting time of $M^p_k$. 
\begin{figure}[ht]
\centering
\epsfig{file=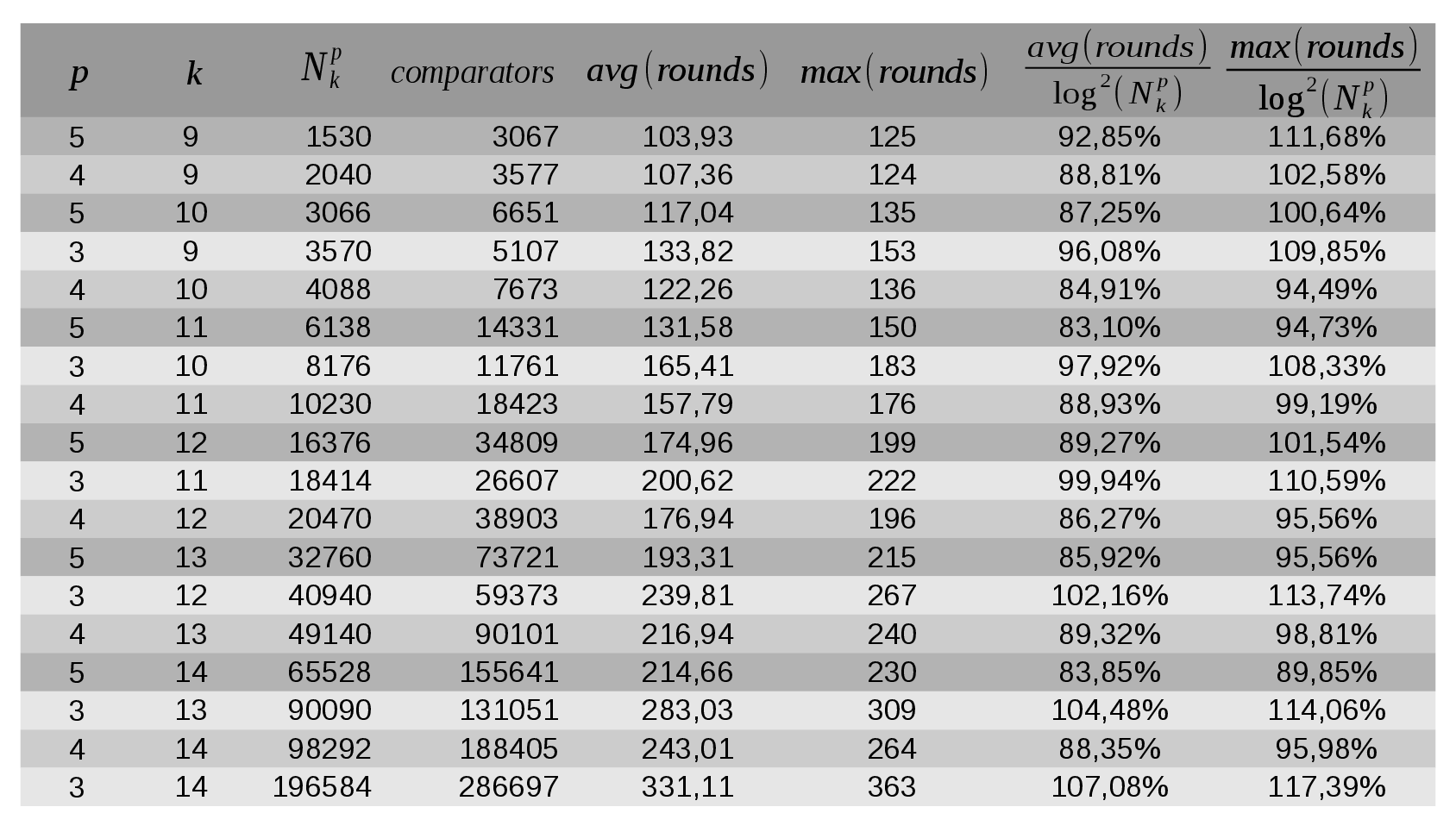, width=4.8in}
\caption{The average sorting times of $10^5$ random permutations with $M^p_k$ 
networks for $3\le p\le 5$ and $9\le k\le 14$.}
\label{sort_times}
\end{figure}

\section{Conclusions}  

For each $k \ge p \ge 4$ we have shown a construction of a p-periodic merging
comparator network of $N^p_k = (2^k-2)\lceil\frac{k-2}{p-2}\rceil$ registers and
proved that it merge any two sorted sequences (given in odd and even registers,
respectively) in time $D^p_k = p(b_k - 1) = p (2\lceil\frac{k-2}{p-2}\rceil - 1)$.
The construction is regular and quite simple. It is created based on the duality
between constant-periodic and constant-delay comparator networks and can be
considered as a natural extension of the previous construction of $3$-periodic
merging networks. Also the proof is a generalisation of the corresponding proof
given for $3$-periodic merging networks. An open question remains whether the
given merging times are optimal for $p$-periodic comparator networks.

Finally, one can observe that for $k > p \ge 4$ we get $2^k \le 2(2^k-2) \le
N^p_k$, which implies $k \le \log N^p_k$. Now we can bound merging times $D^p_k$
for $p = 4,5,6$ as $D^4_k \le 4k-8 \le 4\log N^4_k$, $D^5_k \le 3.33\log N^5_k$
and $D^6_k \le 3\log N^6_k$. Because of skipped negative terms, exact ratios to
$\log N^p_k$ are even better for small values of $k$ (compare Fig.
\ref{merge_times}). 
\begin{figure}[ht]
\centering
\epsfig{file=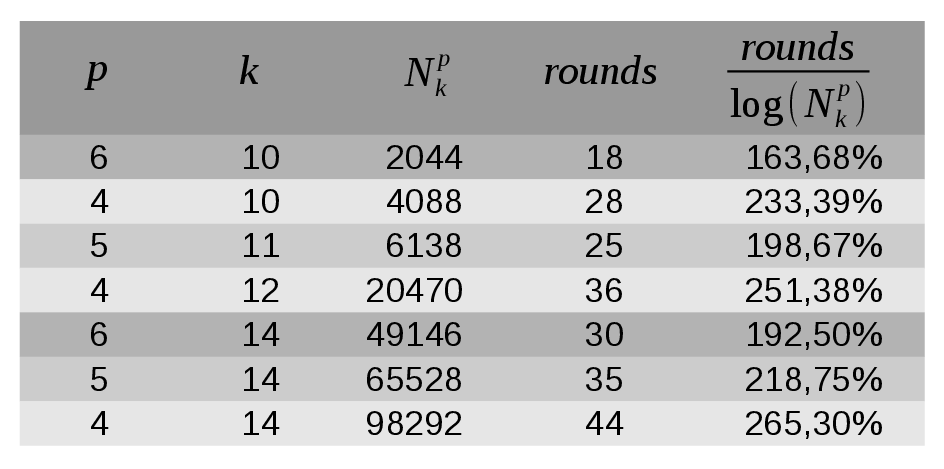, width=2.8in}
\caption{The merging times of some $M^p_k$ periodic
networks compared to $\log N$ non-periodic merging time.}
\label{merge_times}
\end{figure}

\bibliography{merging-spaa}{}

\begin{figure}
\centering
\epsfig{file=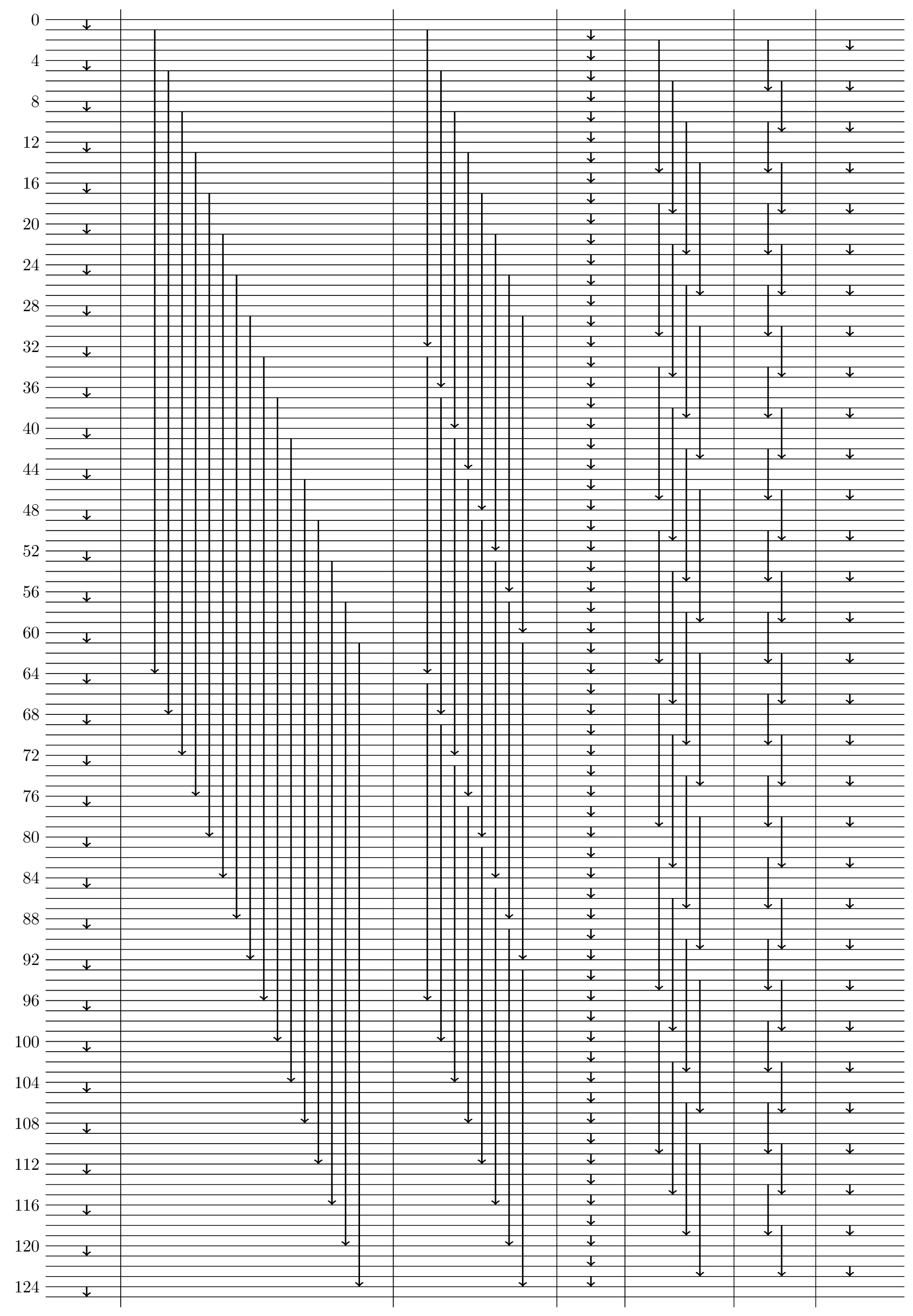, height=3in}
\hspace{20pt}
\epsfig{file=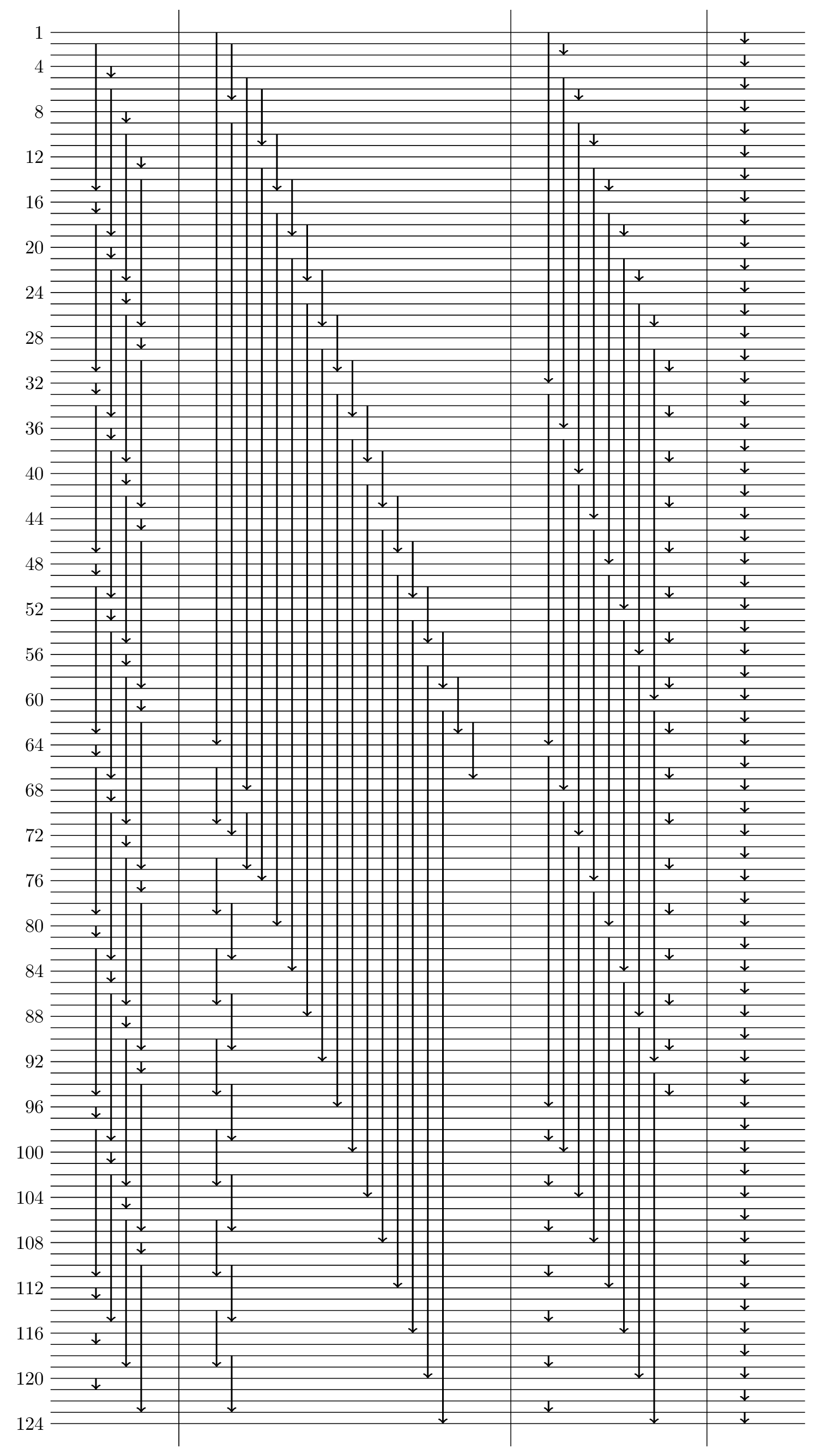, height=3in}
\caption{The traditional drawing of  $P^4_6$ (left) and $M^4_6$ (right) networks. 
Vertical lines separate stages.}
\label{merge46}
\end{figure}

\end{document}